\documentclass[prl,twocolumn,longbibliography,nofootinbib,preprintnumbers,notitlepage,fleqn,superscriptaddress]{revtex4-1} 

\usepackage[tracking=true]{microtype}   

\usepackage[dvipsnames]{xcolor}

\usepackage{color,soul}

\usepackage{verbatim} 

\usepackage{contour}
\usepackage{ulem}

\contourlength{0.8pt}

\newcommand{\myuline}[1]{%
  \uline{\phantom{#1}}%
  \llap{\contour{white}{#1}}%
}
\renewcommand{\emph}{\textit}

\usepackage[bookmarks = true, citecolor = blue, colorlinks = true, linkcolor = magenta, urlcolor = blue]{hyperref}

\usepackage[T1]{fontenc}			
\usepackage[sc,osf]{mathpazo}   	

\usepackage{amsmath}  			
\usepackage{amsfonts}  			
\usepackage{graphicx}   			
\usepackage{mathrsfs, amsthm, amssymb}
\usepackage{bbm, bm}

\usepackage{nicefrac}    			

\usepackage{fourier-orns}  		

\usepackage[braket, qm]{qcircuit}	

\newtheorem{theorem}{Theorem}
\newtheorem{definition}{Definition}
\newtheorem{proposition}{Proposition}
\newtheorem{lemma}{Lemma}
\newtheorem{corollary}{Corollary}
\newtheorem{observation}{Observation}
\newtheorem{claim}{Claim}

\usepackage{tikz}
\usetikzlibrary{
  calc,
  decorations.pathmorphing,
  shapes,
  arrows.meta
}

\newcommand\xlrsquigarrow{%
  \mathrel{%
    \vcenter{%
      \hbox{%
        \begin{tikzpicture}
          \path[
            draw,
            >={Implies[]},
            <->,
            double distance between line centers=1.5pt,
            decorate,
            decoration={
              zigzag,
              amplitude=0.8pt,
              segment length=4pt,
              pre length=4pt,
              post length=4pt,
            },
          ]   
            (0,0) -- (20pt,0);
        \end{tikzpicture}%
      }%
    }%
  }%
}

\usepackage{tikz}
\usetikzlibrary{calc,decorations.pathmorphing,shapes}

\newcounter{sarrow}
\newcommand\xrsquigarrow[1]{%
\stepcounter{sarrow}%
\mathrel{\begin{tikzpicture}[baseline= {( $ (current bounding box.south) + (0,-0.5ex) $ )}]
\node[inner sep=.5ex] (\thesarrow) {$\scriptstyle #1$};
\path[draw,<-,decorate,
  decoration={zigzag,amplitude=0.7pt,segment length=1.2mm,pre=lineto,pre length=4pt}] 
    (\thesarrow.south east) -- (\thesarrow.south west);
\end{tikzpicture}}%
}

\begin{document}


\title[]{
Identical particles as a genuine non-local resource}


\author{Pawel \surname{Blasiak}}
\email{pawel.blasiak@ifj.edu.pl}
\affiliation{Institute for Quantum Studies, Chapman University, Orange, CA 92866, USA}
\affiliation{Institute of Nuclear Physics, Polish Academy of Sciences, 31342 Krak\'ow, Poland}

\author{Marcin \surname{Markiewicz}}
\email{marcinm495@gmail.com}
\affiliation{Institute of Theoretical and Applied Informatics, Polish Academy of Sciences, 44100 Gliwice, Poland}
\affiliation{International Centre for Theory of Quantum Technologies, University of Gda\'nsk, 80308 Gda\'nsk, Poland}


\begin{abstract}
All particles of the same type are indistinguishable, according to a fundamental quantum principle. This entails a description of many-particle states using symmetrised or anti-symmetrised wave functions, which turn out to be formally entangled. However, the measurement of individual particles is hampered by a mode description in the second-quantised theory that masks this entanglement. Is it nonetheless possible to use such states as a resource in Bell-type experiments? More specifically, which states of identical particles can demonstrate non-local correlations in passive linear optical setups that are considered purely classical component of the experiment? Here, the problem is fully solved for multi-particle states with a definite number of identical particles. We show that \textit{all} fermion states and \textit{most} boson states provide a sufficient quantum resource to exhibit non-locality in classical optical setups. The only exception is a special class of boson states that are reducible to a single mode, which turns out to be locally simulable for any passive linear optical experiment. This finding highlights the connection between the basic concept of particle indistinguishability and Bell non-locality, which can be observed by classical means for almost every state of identical particles.


\end{abstract}


\maketitle

\onecolumngrid
\vspace{-0.2cm}

\begin{flushright}
\small\textsl{“No acceptable explanation for the miraculous identity of particles of the same type has ever been put forward.\\That identity must be regarded, not as a triviality, but as a central mystery of physics.”}\vspace{0.1cm}

{-- C.~W.~Misner, K.~S.~Thorne \& J.~A.~Wheeler\\\textit{Gravitation}, 1973.}
\end{flushright}\vspace{0.1cm}



\twocolumngrid

\section{\textsf{Introduction}}\vspace{-0.3cm}

Bell nonlocality stands out as a hallmark of quantum theory, defying our classical conception of reality~\cite{Be64,Be93,BrCaPiScWe14,Sc19}. It boils down to strange correlations between outcomes of space-time separated measurements that \textit{cannot} be explained in the usual causal manner implied by the experimental design, i.e., only as an effect of preparation at the source. 
Quantum theory predicts the appearance of such correlations when the system is prepared in an entangled state, which has been confirmed in a number of ingenious experiments~\cite{AsDaRo82,HeBeDrReKaBlRu15,GiVeWeHaHoPhSt15,ShMeChBiWaStGe15,As15}. Inherent to the notion of entanglement is the concept of a composite system described as the tensor product of subsystems, which presupposes that a full range of measurements on individual subsystems can be made~\cite{Gi91,HoHoHoHo09}. 
This raises the subtle question of whether a given quantum state can actually demonstrate Bell non-locality. The answer will crucially depend on how the subsystems are defined and which measurements are allowed in the experiment. 

This question is particularly significant for identical particles. According to the fundamental principle of indistinguishability, multi-particle systems are described by entangled states, due to the symmetrisation or anti-symmetrisation of the wave function~\cite{FeLeSa65,FeWa71}. However, in the second-quantised theory, this fact is obscured by moving to the mode description of the system where individual particles cannot be directly addressed, thereby preventing explicit probing of subsystems where this entanglement resides. A sensible criterion of entanglement for identical particles is thus challenging~\cite{WiVa03,BeFlFrMa20}, which suggests a shift in focus to a more practical problem of whether and how states of identical particles can result in non-local correlations observable in real experiments. This leads us to deliberately evade considering the notion of {entanglement} for a system of identical particles in favour of {non-local correlations}, since the latter directly relates to what is experimentally observed rather than being a mathematically driven construct.

\begin{figure*}[t]
\begin{center}
\quad\includegraphics[width=1.6\columnwidth]{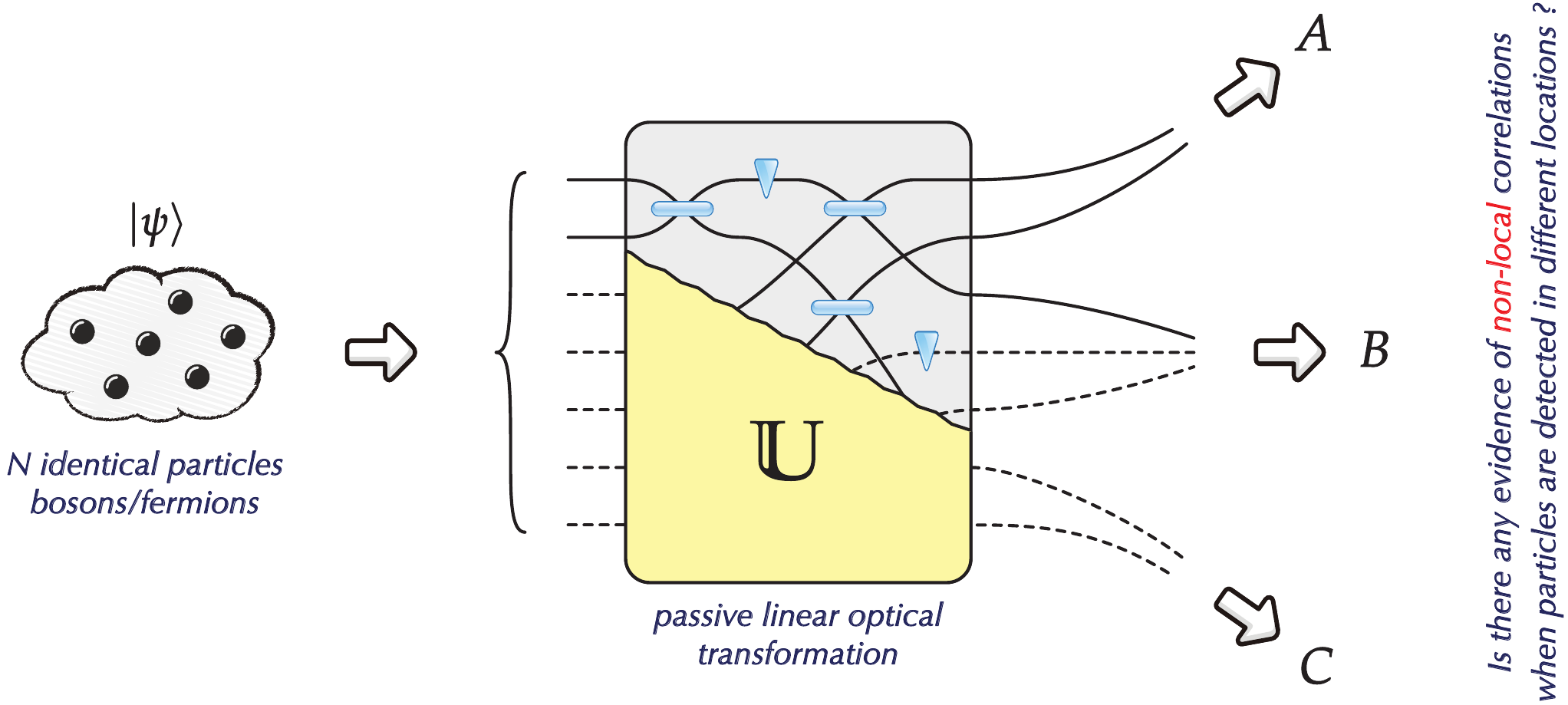}
\caption{\label{Diag_Indistinguishable}\textbf{Passive linear optical experiment.} A quantum state $\ket{\psi}$ of $N$ identical particles (bosons or fermions) enters a classical optical setup built of paths, mirrors, phase shifters, beam splitters, and detectors. Arranged in various configurations, optical elements implement a unitary transformation $\mathbb{U}$ on each particle~\cite{ReZeBeBe94}. Then, the modes/paths are distributed to different locations $A, B, C ...$, where the particles can be further processed and finally detected. Each experiment has the potential to reveal non-local correlations in the observed statistics.\vspace{-0.2cm}}
\end{center}
\end{figure*}

In this work, we treat passive linear optics (comprised of mirrors, phase shifters, beam splitters, and detectors) as a purely classical component of an experiment. From this perspective, any non-classical behaviour observed must have its origin in the remaining part of the experiment, and thus be attributed to the state that is fed into the classical optical setup. See Fig.~\ref{Diag_Indistinguishable} for illustration.
This view provides a straightforward operational way to assess non-classical potential of a given quantum state (with respect to passive linear optics).
Our research problem has a simple expression in the following question. For a given quantum state of identical particles, we ask:
\begin{eqnarray}\nonumber
\hspace{0.2cm}\textit{\parbox{0.89\columnwidth}{\textls[0]{Does there exist some passive linear optical experiment capable of demonstrating Bell non-locality?\vspace{0cm}}}}
\end{eqnarray}

\noindent A positive answer indicates that the state is a genuine \textit{non-local resource}. Conversely, the lack of any such experiment means that the correlations observed in every optical setup can be explained by a generic local hidden variable model. Then, the state \textit{fails} to provide adequate non-local resources. For example, it was shown that every single-particle state, for both bosons and fermions, can be simulated locally in a generic way~\cite{Bl18}.

Note that the essence of such a formulated problem lies in the arbitrary nature of passive linear optical experiments, which can be composed in any possible manner. Furthermore, our query covers every single state within a sufficiently broad class of interest. In this article, the problem is fully solved for multi-particle states with a \textit{definite number} of identical particles for both boson and fermion statistics.
\begin{eqnarray}\nonumber
\textit{\parbox{0.89\columnwidth}{\textls[0]{We show that \myuline{all} fermion states and the vast \myuline{majority} of boson states provide a genuine non-local resource (with respect to passive linear optics), except for a narrow class of states that are reducible to a single mode.}}}
\end{eqnarray}
See Fig.~\ref{Fig-Clasiffication} for illustration. This result reveals an intriguing link between the basic concept of particle identity and the puzzling feature of Bell's non-locality within the framework of classical optical setups.

\begin{figure}[t]
\begin{center}
\quad\includegraphics[width=\columnwidth]{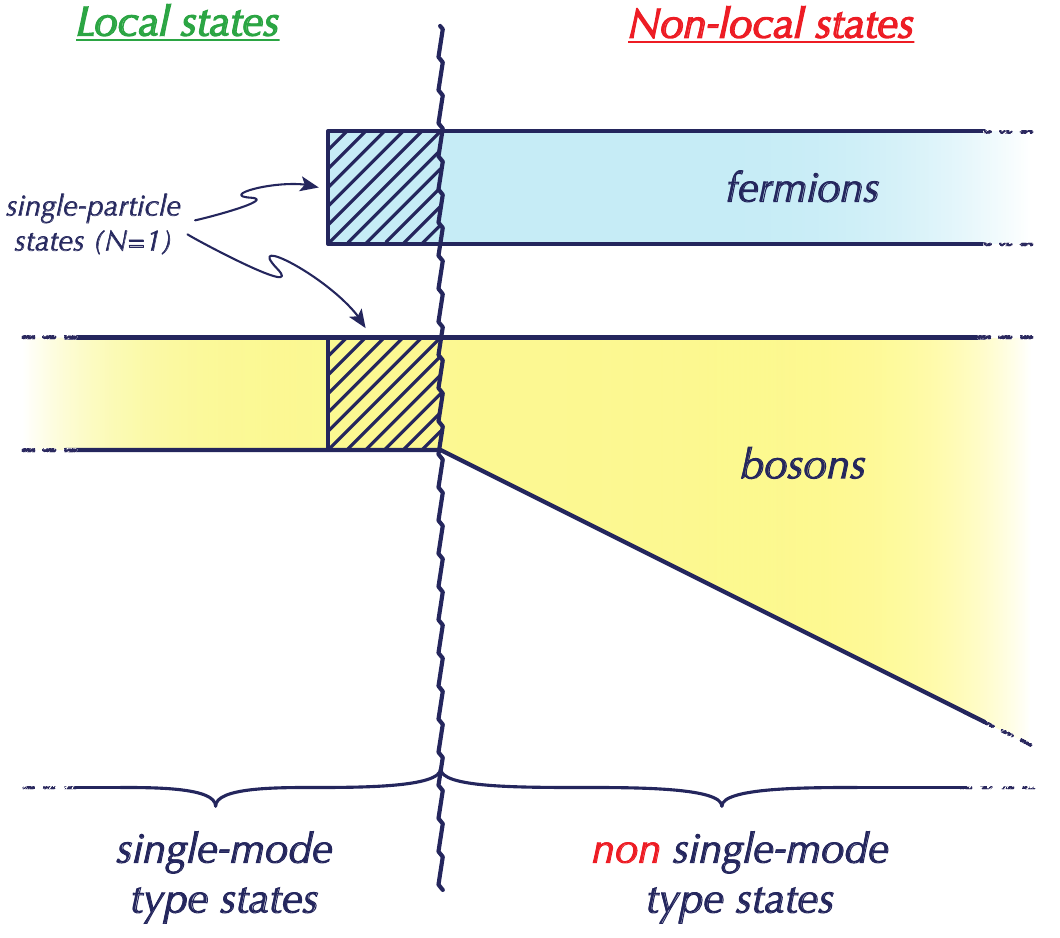}
\caption{\label{Fig-Clasiffication}\textbf{Classification of Bell non-local states (with respect to passive linear optics).} All fermion states provide a non-local resource, except for single-particle states ($N\!=\!1$). The boson statistics has a wider variety of states, which in the local sector includes all states reducible to a single mode. In general, for states with a \textit{definite number} of identical particles, we show that the distinction between \textit{Bell local} vs \textit{Bell non-local} states corresponds to a straightforward division into \textit{single-mode type} vs \textit{non single-mode type} states.\vspace{-0.2cm}}
\end{center}
\end{figure}

\section{\textsf{Results}}\vspace{-0.3cm}

\vspace{0.2cm}\noindent\textbf{\textsf{Informal statement of the problem}}\vspace{0.2cm}

Suppose we lived in a purely classical world and someone presented us with a quantum state of $N$ identical particles (either fermions or bosons). \textit{Would we be able to recognise this new quality with our classical tools?} By classical tools, we mean standard passive linear optics, which includes mirrors, phase shifters, beam splitters, and detectors that can be arranged in various configurations. What can be imagined is that such a quantum state is fed into an optical circuit, processed in any conceivable way, and then the particles are detected in different locations. See Fig.~\ref{Diag_Indistinguishable}. The pattern of correlations observed in such an experiment will be deemed non-classical only if they \textit{cannot} be explained by a local model with particles (and/or hidden variables) propagating along the paths as specified by the experimental design. Then, the state can be called a genuine \textit{non-local (or non-classical) resource}. Otherwise, there is an account of every optical experiment that does not require any non-local mechanisms, and therefore the state should not give rise to surprise.

The problem is formulated by thinking of the quantum state as a potential \textit{non-classical resource} with respect to a fixed set of operations (i.e., passive linear optics and detectors).
Specifically, we are interested in the following question:
\begin{eqnarray}\nonumber
\textit{\parbox{0.89\columnwidth}{\textls[0]{Which quantum states of identical particles can lead to non-local behaviour in classical experiments composed of passive linear optical transformations and detectors?}}}
\end{eqnarray}
For each particular state, it has two possible answers:
\begin{itemize}
\item[\textsc{Yes:}]{The state is a genuine \textit{non-local resource}. It means that the state can be used in some experimental procedure to demonstrate non-local correlations in space-time separated regions (by violating some Bell-type inequality in some specific experiment).}
\item[\textsc{No:}]{The state does \textit{not} provide a \textit{non-local resource}. That is, the correlations observed in any conceivable experiment can be explained in a local manner (via some generic model with hidden variables propagating locally following the experimental design).}
\end{itemize}

\noindent This splits the space of quantum states into two disjoint classes. The first class can fuel classical optical setups to operate at a non-local level. The second class appears to be classical-like, i.e. simulable with local resources, and hence one should not expect a non-classical advantage in passive linear optics with those states.

Here, we are interested in the classification of all quantum states with a \textit{definite} number of particles ($N$ is fixed). Notably, for a single particle ($N\!=\!1$), the question was resolved in Ref.~\cite{Bl18}, where a generic local hidden variable model of arbitrary passive linear optical experiments was constructed (i.e., single-particle states do not provide a non-local resource). However, a far more complex and interesting problem is that of multi-particle systems ($N\!\geqslant\!2$). We show that \textit{all} fermion states and \textit{most} boson states provide a genuine \textit{non-local resource}. 

\vspace{0.4cm}\noindent\textbf{\textsf{Passive linear optical toolkit}}\vspace{0.2cm}

In the \textit{occupation number representation} (\textit{Fock space}), a general state of $N$ identical particles is described by
\begin{eqnarray}\label{N-particle-state-psi}
\ket{\psi}&=&\sum_{n_1+\,...\,+n_M=N}\psi_{n_1\,...\,n_M}\ket{n_1\,...\,n_M},
\end{eqnarray}
which is as a (complex) combination of states $\ket{n_1\,...\,n_M}$ with a definite number of particles $n_i$ in the respective modes $i=1,...\,,M$. Those numbers are constrained by the particle statistics: for \textit{bosons} $n_i\in\mathbb{N}$, and for \textit{fermions} $n_i=0,1$. In the following, we will consider situations where the total number of particles $N$ is \textit{fixed} and \textit{without} any restriction on the number of modes $M$. Furthermore, since we are concerned with locality, we will focus on optical modes of the system, i.e., think of each mode as a different path along which particles can propagate.

It will be convenient to use the standard representation in terms of the \textit{creation}/\textit{annihilation} operators, which add/subtract a particle in the respective mode of the system. Then Eq.~(\ref{N-particle-state-psi}) can be written the form
\begin{eqnarray}\label{N-particle-state-adag}
\ket{\psi}&=&\sum_{n_1+\,...\,+n_M=N}\psi_{n_1\,...\,n_M}\ \prod_{i=1}^M\,\frac{{a_i^\dag}^{n_i}}{\sqrt{n_i!}}\ket{0},
\end{eqnarray}
where $\ket{0}$ is the \textit{vacuum} state and the \textit{creation} and \textit{annihilation} operators, $a^\dag_i$ and $a_i$, satisfy the usual commutation relations specific for the respective \textit{boson} or \textit{fermion} statistics~\cite{FeLeSa65,FeWa71}.\footnote{For reference, we note that the conventional commutation relations are defined as follows:  $[a_i,a_j]_{\mp}\!=\![a_i^\dag,a_j^\dag]_{\mp}\!=\!0\,,\,[a_i,a_j^\dag]_{\mp}\!=\!\delta_{ij}$, where the difference between the \textit{boson/fermion} statistics boils down to the use of the \textit{commutation/anti-commutation} relations ($\mp$). We use their canonical representation whose action in a given mode for \textit{bosons} is given by $a_i^\dag\ket{n_i}\!=\!\sqrt{n_i+1}\ket{n_i}$ and $a_i\ket{\,n_i}\!=\!\sqrt{n_i}\ket{n_i-1}$, while for \textit{fermions} we have $a_i^\dag\ket{0_i}\!=\!\ket{1_i}$, $a_i\ket{\,1_i}\!=\!\ket{0_i}$ and $a_i^\dag\ket{1_i}\!=\!a_i\ket{\,0_i}\!=\!0$. This complies with the restriction on the occupation numbers for \textit{bosons} $n_i\in\mathbb{N}$, and for \textit{fermions} $n_i=0,1$. Then, the \textit{number states} can be written in a compact form $\ket{n_1\,...\,n_M}\ =\ \frac{1}{\sqrt{n_1!\,...\,n_M!}}\ a_{  1}^{  \dag\, n_1}...\,a_{  M}^{  \dag\, n_M}\ket{0}$. See Refs.~\cite{FeLeSa65,FeWa71} for a full exposition.}

In the following, we will be interested in transformations of the Fock states that can be implemented with passive linear optics, i.e., a sequence of single-mode and two-mode operations (or gates) that are conventionally realised by mirrors, phase shifters and beam splitters. Such transformations are described as follows
\begin{eqnarray}\label{U-a-dag}
a_i^\dag&\xymatrix{\ar[r]^{\atop \mathbb{U}} &}& {a'\!}_i^{\,\,\dag}\ =\ \sum_{  j=1}^{  M}{U}_{ij}\,a^{  \dag}_{  j},
\end{eqnarray}
for each mode $i=1,\,...\,,M$ of the system, and $\mathbb{U}\equiv [U_{ij}]$ is a unitary transformation on $\mathbb{C}^M$. Furthermore, it is known that any unitary $\mathbb{U}$ on a single-particle space $\mathbb{C}^M$ can be realised by passive linear optics (i.e., can be obtained as a sequence of single- and two-mode operations implemented by mirrors, phase shifters and beam splitters)~\cite{ReZeBeBe94}. However, this is not enough to generate any transformation on the multi-particle Fock space since under passive linear optics, the state space splits into a myriad of non-equivalent continuous classes whose full classification is beyond reach, see Ref.~\cite{MiRoOsLe14}. 

For further convenience, let us formally distinguish a special class of states that will play an important role in the following discussion.

\begin{definition}\label{Def-single-mode-type}\ \\
A state is said to be of a \textbf{{single-mode type}} if it can be reduced to a single mode by passive linear optics. That is, for a state in Eq.~(\ref{N-particle-state-psi})/(\ref{N-particle-state-adag}), it means that there exists a unitary $\mathbb{U}$ on $\mathbb{C}^M$ (or a sequence of optical elements) such that
\begin{eqnarray}\label{single-mode-def}
\ket{\psi}&=&\mathbb{U}\ket{N,0,...\,,0}\,=\,\frac{{{a'\!}_1^{\,\,\dag}}^{N}}{\sqrt{N!}}\ket{0}.
\end{eqnarray}
\end{definition}
Equivalently, those states can be thought of as having evolved from the initial state in which all the particles were in the same mode. 
A straightforward consequence of Ref.~\cite{ReZeBeBe94} is that all single-particle states ($N\!=\!1$) are of a single-mode type (for both bosons and fermions). Clearly, in the multi-particle case ($N\!\geqslant\!2$), the class of single-mode type states is nontrivial only for bosons (it is empty for fermions due to the Pauli exclusion principle). We hasten to note that \textbf{Definition~\ref{Def-single-mode-type}} boils down to a few easy-to-check conditions; see \textbf{Supplementary Information} for a discussion.

\vspace{0.4cm}
\noindent{\textbf{\textsf{Bell non-locality in optical experiments}}}\vspace{0.2cm}

Every experiment takes place in space and time. This is particularly manifest in optical setups that are composed of basic elements distributed in different positions and interconnected by optical paths into a larger circuit. In a typical optical experiment, an initial state $\ket{\psi}$ is fed into a circuit that is responsible for evolving the system and registering the particles at specific locations. Here, without the loss of generality, the basic components of passive linear optical setups can be considered as single- and two-mode gates, such as mirrors, phase shifters, beam splitters and number-resolving detectors~\cite{ReZeBeBe94}. 

The spatio-temporal arrangement of optical designs has implications for the possible interpretations of the experiment. Namely, it is common to think that the flow of information follows the pattern of interconnections in the circuit. That is, its physical carriers propagate along the paths of the circuit (whether they are particles, fields or hidden variables), and the information is processed in a modular manner. This view stems from the principle of local causality, which is especially appealing for space-like separated arrangements. See Fig.~\ref{Diag_Indistinguishable}. It was original Bell's insight that quantum correlations can be at odds with the causal structure implied by the experimental design~\cite{Be64,Be93,BrCaPiScWe14,Sc19}. In the optical framework, this idea has the following expression.

\begin{definition}\label{non-local-behaviour}\ \\
An optical experiment demonstrates \textbf{non-local behaviour} (or \textbf{Bell non-locality}) if the observed correlations \myuline{cannot} be explained in a cause-and-effect manner following the pattern of interconnections in the experimental design.
\end{definition}

In other words, there is no local hidden variable model explaining the observed experimental statistics, where locality requires that the variables propagate from one optical element to another along the paths determined by the circuit. In practice, demonstrating non-locality in a given experiment boils down to a violation of some Bell-type inequality derived from the corresponding causal structure. The opposite statement can only be justified by constructing an explicit local hidden variable model that is compatible with the experimental setup.


\vspace{0.4cm}
\noindent{\textbf{\textsf{Main result}}}\vspace{0.2cm}

An optical experiment can be seen as a test of certain properties of the state that is supplied to the circuit. In \textbf{Definition~\ref{non-local-behaviour}}, the question of non-locality is posed for a given experimental setup, which includes both the state and the optical design. In this paper, we are interested in the property of the state \textit{itself} without restricting it to a particular implementation. We ask whether the state can demonstrate Bell non-locality by admitting a wider range of possible arrangements. This idea comes from thinking about the state as a resource of non-locality in a larger class of experimental designs. Here we focus on the use of passive linear optics, as defined below.

\begin{definition}\label{non-local-resource}\ \\
A state $\ket{\psi}$ is a genuine \textbf{non-local resource} with respect to passive linear optics if it is capable of manifesting non-local behaviour in \myuline{some} passive linear optical experiment.
\end{definition}

Note that to demonstrate non-locality of a given state, it is sufficient to provide a \textit{single} example of a passive linear optical circuit in which Bell non-locality can be observed. However, showing the opposite for a given state requires constructing a generic local hidden variable model explaining correlations obtained in \textit{every} conceivable experiment with passive linear optical setups.

Our main result takes the form of a simple criterion for states with a \textit{definite number} of particles in Eq.~(\ref{N-particle-state-psi})/(\ref{N-particle-state-adag}).

\begin{theorem}\label{theorem}\ \\
A state $\ket{\psi}$ is a  \textbf{non-local resource} with respect to passive linear optics if, and only if, it is \myuline{not} of the \textbf{single-mode type}.
\end{theorem}

This provides a straightforward classification of multi-particle states based on a simple condition \textbf{Definition~\ref{Def-single-mode-type}} and easy-to-understand interpretation. Interestingly, the only states that lack the non-local potential are those originating in the same mode. This leads to the following characteristics of states for both boson and fermion statistics; see illustration in Fig.~\ref{Fig-Clasiffication}. 

\begin{corollary}\ \\
Neither fermions nor bosons provide a non-local resource for \myuline{any} single-particle state ($N\!\!=\!\!1$). For the multi-particle case ($N\!\geqslant\!2$), \myuline{all} fermion states are non-local resources, while \myuline{some} boson states are not (i.e., those of a single-mode type). 
\end{corollary}



We refer to the \textbf{Methods} section for the detailed proof of \textbf{Theorem~\ref{theorem}}. Let us conclude by unfolding a partial result which facilitates the full proof. 

\vspace{0.4cm}\noindent{\textbf{\textsf{Yurke-Stoler test: A useful lemma}}}\vspace{0.2cm}

The key tool in our discussion of non-locality in the interferometric setups is the co-called \textit{Yurke-Stoler test}. It is designed for the special case of two particles ($N\!=\!2$) and two modes ($M\!=\!2$), and asks about Bell non-locality of the following general state 
\begin{eqnarray}\label{Psi-M2-N2}
\ket{\phi}&=&\Big(\,\alpha\ \tfrac{{a_{1}^\dag}^2}{\sqrt{2!}}\,+\,\beta\ a_{1}^\dag\,a_{2}^\dag\,+\,\gamma\ \tfrac{{a_{2}^\dag}^2}{\sqrt{2!}}\,\Big)\ket{0},
\end{eqnarray}
where $|\alpha|^2+|\beta|^2+|\gamma|^2=1$ in the boson case, whilst for fermions we have $|\beta|^2=1$, $\alpha=\gamma=0$. In the following, we show how the seminal idea in Ref.~\cite{YuSt92} can be refined to serve the purpose at hand.

\begin{figure}[t]
\centering
\includegraphics[width=\columnwidth]{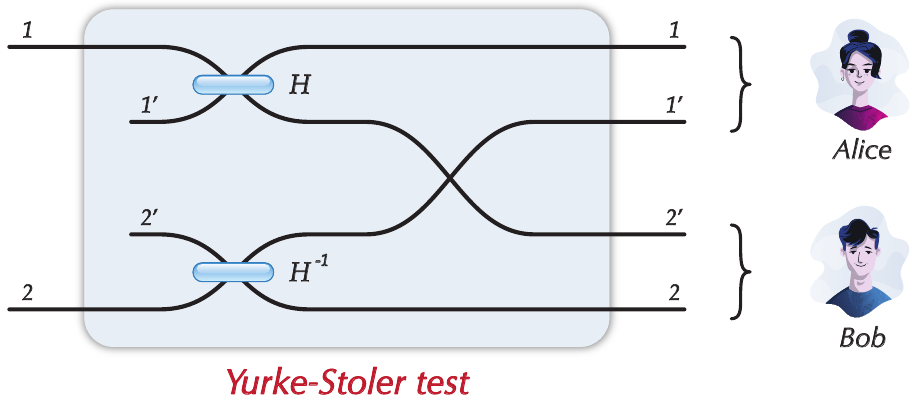}
\caption{\label{Fig-YS}{\bf\textsf{\mbox{Yurke-Stoler test.}}} A two-particle state, Eq.\,(\ref{Psi-M2-N2}), in modes 1 \& 2 is split by two beam splitters (Hadamard gates) into two dual-rail qubits \{\text{1},1'\} and \{2,2'\}. Then the modes 1' and 2' get swapped and Alice \& Bob make the usual Bell test on their (dual-rail) qubits when each one of them gets a single particle.\vspace{-0.2cm}}
\end{figure}

\textit{Yurke-Stoler} test checks for a violation of Bell inequalities in the design depicted in Fig.~\ref{Fig-YS}. It takes the state Eq.~(\ref{Psi-M2-N2}) in two input modes 1 \& 2, that are split into two \textit{dual-rail qubits} \{\text{1},1'\} and \{2,2'\}. Then the modes 1' and 2' get swapped, and the dual-rail qubits are sent out to Alice and Bob, who perform a typical Bell experiment. We note that the choice of settings on Alice and Bob's side is unrestricted since any projective measurement on their dual-rail qubits can be implemented with passive linear optics. The protocol involves \textit{post-selection}, which retains only those experimental runs in which a single particle is detected on Alice and Bob's side (i.e., when the dual-rail qubits are well-defined). Notably, the specifics of post-selection used in the experiment do not compromise the conclusions from the violation of Bell inequalities when the number of particles is conserved, as shown in~\cite{PoHaZu97,BlBoMa21}. It means that the protocol provides a genuine test of non-local correlations in the system.

This scheme was originally proposed by Yurke and Stoler to prove Bell non-locality for two particles coming from independent sources~\cite{YuSt92}. For a recent discussion see~\cite{BlMa19}. Here, we build on this idea taking it as a test of Bell non-locality for a general two-particle state $\ket{\phi}$ in Eq.~(\ref{Psi-M2-N2}), which via the setup in Fig.~\ref{Fig-YS} evolves as follows
\begin{widetext}
\begin{eqnarray}
\ket{\phi}&\xymatrix{\ar[r]^{\atop H,H^{ \text{-1}}}&}&\Big(\,\tfrac{\alpha\,}{2}\tfrac{(a_{1}^\dag\,+\,a_{1'}^\dag)^2}{\sqrt{2!}}\,+\,\tfrac{\beta}{2}\,(a_{1}^\dag+a_{1'}^\dag)\,(a_{2}^\dag\,+\,a_{2'}^\dag)\,+\,\tfrac{\gamma}{2}\tfrac{(a_{2}^\dag\,+\,a_{2'}^\dag)^2}{\sqrt{2!}}\,\Big)\ket{0}
\\
&\xymatrix{\ar[r]^{\atop 1'\leftrightarrow\,2'}&}&\Big(\,\tfrac{\alpha\,}{2}\tfrac{(a_{1}^\dag\,+\,a_{2'}^\dag)^2}{\sqrt{2!}}\,+\,\tfrac{\beta}{2}\,(a_{1}^\dag\,+\,a_{2'}^\dag)\,(a_{2}^\dag+a_{1'}^\dag)\,+\,\tfrac{\gamma}{2}\tfrac{(a_{2}^\dag\,+\,a_{1'}^\dag)^2}{\sqrt{2!}}\,\Big)\ket{0}
\\
&\xrsquigarrow{\text{\!\tiny{\emph{post-select}}\!}}&\Big(\,\tfrac{\alpha}{\sqrt{2!}}\ a_{1}^\dag\,a_{2'}^\dag\,+\,\tfrac{\beta}{2}\ a_{1}^\dag\,a_{2}^\dag
\,\pm\,\tfrac{\beta}{2}\ a_{1'}^\dag\,a_{2'}^\dag
\,+\,\tfrac{\gamma}{\sqrt{2!}}\ a_{1'}^\dag\,a_{2}^\dag\,\Big)\ket{0},
\end{eqnarray}
\end{widetext}
where the "$-$" sign refers to fermions that anti-commute (for which we also have $|\beta|^2=1$, $\alpha=\gamma=0$). Note that in the last line only the terms with a single particle in each dual-rail qubit \{\text{1},1'\} and \{2,2'\} are retained. Clearly, this state is unnormalised due to post-selection, which succeeds with probability \nicefrac{1}{2}. 
As a result, Alice and Bob share a dual-rail encoded two-qubit state in the form
\begin{eqnarray}
\ket{\chi}&\sim&\tfrac{\alpha}{\sqrt{2!}}\ket{\uparrow\uparrow}+\tfrac{\beta}{2}\ket{\uparrow\downarrow}\pm\tfrac{\beta}{2}\ket{\downarrow\uparrow}+\tfrac{\gamma}{\sqrt{2!}}\ket{\downarrow\downarrow},
\end{eqnarray}
where $\ket{\uparrow}\equiv\ket{10}$ and $\ket{\downarrow}\equiv\ket{01}$ is the computational basis for the respective dual-rail qubit \{\text{1},1'\} and \{2,2'\}.

Alice and Bob then conduct the standard Bell test with an arbitrary choice of measurements on either side. It will only reveal non-local correlations for some choice of settings if, and only if, $\ket{\chi}$ is an entangled state~\cite{Gi91}. Conversely, the lack of non-local correlations observed in the Yurke-Stoler test requires that it must be a product state, i.e., we need to have $\beta^2=2\,\alpha\gamma$.\footnote{It is straightforward to see that the state $e\ket{\uparrow\uparrow}+f\ket{\uparrow\downarrow}+g\ket{\downarrow\uparrow}+h\ket{\downarrow\downarrow}$ is a product state if, and only if, $eh=fg$.} 
Note that the latter condition implies that Eq.~(\ref{Psi-M2-N2}) simplifies to the single-mode type form (possible only for bosons)
\begin{eqnarray}
\ket{\phi}&=&
\tfrac{1}{\sqrt{2!}}\left(\sqrt{\alpha}\,a_1^\dag+\sqrt{\gamma}\,a_2^\dag\right)^2\ket{0}\,=\,\frac{{{a'\!}_1^{\,\,\dag}}^{2}}{\sqrt{2!}}\ket{0}.
\end{eqnarray}
This proves a helpful result that can be used to analyse more complex interferometric designs.
\begin{lemma}\label{lemma}\ \\
The state $\ket{\phi}$ in Eq.~(\ref{Psi-M2-N2}) does \myuline{not} manifest non-local correlations in the Yurke-Stoler test if, and only if, it is of a single-mode type, i.e. when we have $\beta^2=2\,\alpha\gamma$.
\end{lemma}

Note that for fermions, the Yurke-Stoler test is enough to demonstrate Bell non-locality for \textit{every} two-particle ($N\!=\!2$) and two-mode ($M\!=\!2$) state. However, for the boson case, this setup is insufficient to prove non-locality for the single-mode type states. It will turn out that those states are locally simulable for any passive linear optical circuit. See \textbf{Methods} for the generic construction.






\section{\textsf{Discussion}}\vspace{-0.3cm}

The main interest of this work is to explore the boundary between the classical and quantum features of identical particles. We approach this task by clearly distinguishing between the classical and quantum parts of the experimental setups. This allows us to ask the question of whether and when the quantum state introduces a new quality that cannot be explained by classical means. In other words, can the quantum state \textit{itself} be thought of as a genuine non-local resource? For this reason, in the classical part we include only standard optical means of processing the state, consisting of passive linear optics and single-mode detection. Note that if any additional quantum components were allowed, then the answer would refer to both the state \textit{and} the extra quantum resources used in the experiment. However, this is a very {different} question. Therefore, when seeking an unambiguous answer, one should avoid using any additional resources, such as extra quantum particles or ancillary states. In particular, this excludes POVMs or homodyne detection from consideration (if these were included, the conclusion of non-classicality would refer jointly to the state \textit{and} the ancillary resource used in the experiment, see e.g.~\cite{Ha94,DuVe07,FuTaZwWiFu15,DaKaMaMaWoZu21}). In short, we stand on the position that adding an ancilla introduces extra quantum particles being a potential source of non-local effects, which undermines the conclusions regarding the non-locality of the studied state \textit{itself} in such extended experimental procedures.

Our result shows a fundamental difference between the system of particles that evolved from a single mode and those that evolved from separate modes. We coined the simple term \textit{single-mode type} state, which turns out to fully characterise the class of states incapable of exhibiting non-local effects in any passive linear optical experiment. It was shown that, for all those states, there exists a generic local hidden variable model that explains the observed statistics. The constructed model can be viewed as the maximum achievable as regards local simulation of multi-particle quantum states in passive linear optics. On the other hand, we have also shown that for every \textit{non-single-mode type} state, a passive linear optical experiment can be found that demonstrates Bell non-locality. This completes the characterisation of non-local states with a \textit{definite number} of identical particles, which boils down to a few easy-to-check conditions defining single-mode type states. We note that our result considerably extends the discussion in Ref.~\cite{Bl18}, where only the single-particle case was discussed.

The derived classification of non-local resources applies to both bosons and fermions, as illustrated in Fig.~\ref{Fig-Clasiffication}. Interestingly, all multiparticle fermion states provide a non-local resource for passive linear optics, whereas for bosons there is a non-trivial class of single-mode type states that lack this property. This difference can is due to the Pauli exclusion principle, which rules out their single-mode origin. We note that, in the non-local sector, bosons admit a wider range of states which exhibit peculiar properties. As a curious example, take the boson sampling protocol designed for passive linear optics~\cite{AaAr13}. It requires non-local states, but not of the type admitted by the fermion statistics. Let it serve as an illustration of the intricate relationship between the fundamental concept of Bell non-locality and complexity in quantum information processing for identical particles. The latter aspect goes beyond the scope of the present paper; see Refs.~\cite{LoCo18,MoYaFaZiTrAd20} for a discussion of those topics.


 
Finally, two general comments should be made about the physics of the research problem. Firstly, this work is not about the property of entanglement, which is a theoretical concept defined within the mathematical formalism of quantum theory. It is about the property called Bell non-locality, which concerns causal explanations of correlations observed in certain experimental setups that are crucially related to their spatial design. Therefore, the question of physical realisations plays a central role in our discussion (which involves passive linear optics and single-mode detection). This attitude is also motivated by the fact that the definition of entanglement for a system of identical particles is challenging due to the mode description in the second-quantised theory. See Ref.~\cite{BeFlFrMa20} for a recent review of the competing approaches, none of which rely on the causal notion of Bell non-locality discussed in this paper. In this sense, the problem concerns the physical concepts and their interpretation, rather than 
the mathematical aspects of quantum formalism. Secondly, we are interested in systems of identical particles (bosons or fermions) that are closer to physical realisations, since any implementation necessarily involves particles that are deemed fundamentally indistinguishable. Although it may seem like a step backwards compared to the modern information-theoretic accounts, we believe this is a step in the right direction, following the slogan \textit{"information is physical"}. This line of research highlights the fundamental quantum principle of indistinguishability of particles, which is sometimes overlooked despite being a core tenet of quantum mechanics (cf. the quote from C. W. Misner, K. S. Thorne \& J. A. Wheeler in \textit{Gravitation}, 1973).

\newpage

\section{\textsf{Methods}}\vspace{-0.3cm}

It is convenient to rephrase \textbf{Theorem~\ref{theorem}} in the equivalent contrapositive form. Using \textbf{Definition~\ref{non-local-resource}}, it unfolds as follows.
\renewcommand\thetheorem{1'}
\begin{theorem}\label{theorem-prime}\ \\
A state $\ket{\psi}$ is locally simulable in every passive linear optical experiment if, and only if, it is of a single-mode type.
\end{theorem}
\noindent This is equivalent to the following two assertions.
\begin{claim}[Necessity]\label{claim1}
A state which does \myuline{not} manifest non-local correlations in any passive linear optical experiment is necessarily of the single-mode type. 
\end{claim}
\begin{claim}[Sufficiency]\label{claim2}
All single-mode type states are locally simulable in arbitrary passive linear optical experiments.
\end{claim}

We note that in Ref.~\cite{Bl18}, it was demonstrated that the single-particle case ($N\!=\!1$) holds true (in that case all states are of the single-mode type~\cite{ReZeBeBe94}). See also the proof of \textbf{Claim~\ref{claim2}} below, where a generic local model for all single-mode type states is given, which includes the single-particle case.

In the following, we will focus on the multi-particle case $N\!\geqslant\!2$. We first prove the easier fermion case and then proceed to sketch the proof for bosons.

\begin{center}\myuline{\textbf{\textsf{Fermion case}}}\end{center}
\vspace{-0.1cm}

There are no multi-particle states of the single-mode type for fermions, due to the Pauli exclusion principle. 
It is thus sufficient to indicate, for each fermion state, an experiment that shows non-local correlations. 

We consider a state in Eq.~(\ref{N-particle-state-psi}) with all $n_i=0,1$. Since 
$N\!\geqslant\!2$, there must exist a pair of indices $k,l$ such that 
$\psi_{n_1\,...\,n_M}\neq0$ for $n_k,n_l\!=\!1$ and some choice of the remaining $n_j$'s. Without loss of generality, we take $k\!=\!1,l\!=\!2$ and denote this particular choice of the remaining indices $\tilde{n}_j$ for $j\geqslant3$. That is, we can assume that we have
\begin{eqnarray}\label{alpha-neq-0}
\psi_{11\,\tilde{n}_3...\,\tilde{n}_M}\neq0\,.
\end{eqnarray}

Now, consider the experimental procedure depicted in Fig.~\ref{Fig-YS-Fermions}. This is an event-ready scheme in which the state $\ket{11}$ in modes 1 \& 2 is heralded by the joint detection of $\tilde{n}_j$ particles in the respective modes $j=3,...\,,M$. The condition Eq.\,(\ref{alpha-neq-0}) guarantees that this happens sometimes. Importantly, for fermions, there is only a single term with occupancy $\tilde{n}_3...\,\tilde{n}_M$ in the expression in Eq.~(\ref{N-particle-state-psi}). The experiment ends with the Yurke-Stoler test carried out on such a prepared state in modes 1 \& 2, which demonstrates non-local correlations. See Fig.~\ref{Fig-YS} and \textbf{Lemma~\ref{lemma}}.

\begin{figure}[t]
\centering
\includegraphics[width=\columnwidth]{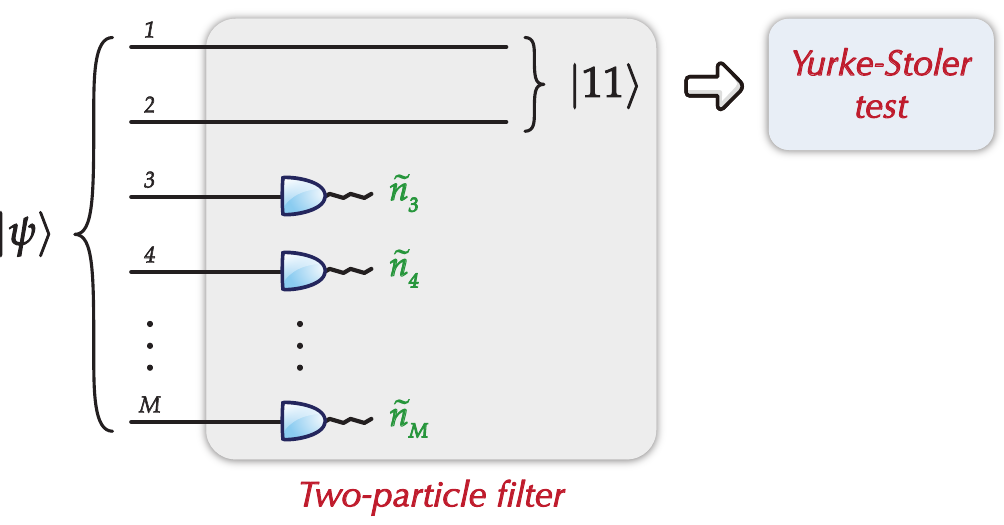}
\caption{\label{Fig-YS-Fermions}{\bf\textsf{\mbox{Filtering out two-particle state in the  fermion case.}}}\\ Conditioning on the detection of $\tilde{n}_j$ particles in the array of detectors $j=3,...\,,M$ prepares the event-ready state $\ket{11}$ in modes 1 \& 2. Then the original Yurke-Stoler test can be made.\vspace{-0.2cm}}
\end{figure}


\begin{center}\myuline{\textbf{\textsf{Boson case}}}
\end{center}
\vspace{-0.1cm}

The boson statistics admits a wider variety of multi-particle states with more subtle correlation effects. This makes the above simple reasoning used for fermions not applicable in the boson case. In the following, we explain the main steps involved in the proof of \textbf{Claim~\ref{claim1}}~and~\textbf{\ref{claim2}}. See \textbf{Supplementary Information} for the details. 

It is crucial to observe that the coefficients in Eq.~(\ref{N-particle-state-psi}), for the single-mode type state can always be expressed in the product form\footnote{We use the abbreviated notation for the \textit{multinomial coefficient}\\$\left(\!{N\atop \vec{n}}\!\right)=\tfrac{N!}{n_1!\,n_2!\,...\,n_M!}$, and by convention $\left(\!{N\atop \vec{n}}\!\right)=0$ for $|\vec{n}|:=\sum_i n_i\neq N$.}
\begin{eqnarray}\label{single-mode-coeff-product}
\psi_{n_1\,...\,n_M}&=&\binom{N}{\vec{n}}^{\!\nicefrac{1}{2}}\,\prod_{i=1}^M\,U_i^{n_i},
\end{eqnarray}
for some complex $U_1,...\,,U_M$ that normalise to one, i.e. ${\sum}_{i=1}^M\,|U_i|^2=1$. This is a simple consequence of the multinomial expansion applied to \textbf{Definition~\ref{Def-single-mode-type}} and Eq.~(\ref{U-a-dag}).
\vspace{0.1cm}

\myuline{\textbf{\textit{Sketch of proof of \textbf{Claim~\ref{claim1}}}}}\ \vspace{0.2cm}

The idea of the proof is to consider examples of passive linear optical experiments and require that they \textit{do not} exhibit non-local correlations. This will imply certain constraints on the coefficients in Eq.~(\ref{N-particle-state-psi}), necessitating the form of Eq.~(\ref{single-mode-coeff-product}). The key tool in the analysis is the Yurke-Stoler test and \textbf{Lemma~\ref{lemma}} discussed in the main text. It proves the case for two particles and two modes ($N\!=\!2, M\!=\!2$). In the following, this result will be used to derive constraints in certain experimental setups.
\vspace{0.2cm}

\textsf{\textit{$\bullet$ \myuline{Many particles \& two modes ($N\!\geqslant\!2, M\!=\!2$)}}}\vspace{0.1cm}

A general two-mode state with arbitrary number of particles $N$ can be written in form  
\begin{eqnarray}\label{psi-M2}
\ket{\phi}&=&\sum_{n=0}^N\beta_n\,\tfrac{{a_1^\dag}^n}{\sqrt{n!}}\tfrac{{a_2^\dag}^{N-n}}{\sqrt{(N-n)!}}\ket{0}.
\end{eqnarray}

\begin{figure}[t]
\centering
\includegraphics[width=\columnwidth]{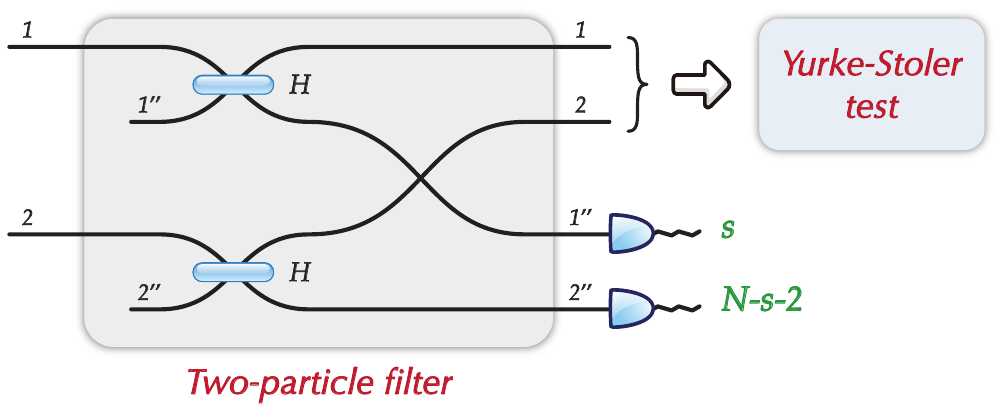}
\caption{\label{Fig-YS-M2}{\bf\textsf{\mbox{Yurke-Stoler type experiment with two-particle filter.}}}\\ In the initial phase $N-2$ particles are pulled out from the system by conditioning on $s$ and $N-s-2$ particle detections in the auxiliary modes 1'' \& 2''. For each $s=0,1,...\,,N-2$\,, this guarantees event-ready preparation of some two-particle state in modes 1 \& 2 which then undergo the Yurke-Stoler test.\vspace{-0.2cm}}
\end{figure}

Let us start by considering a \textit{two-particle filter}, shown in Fig.~\ref{Fig-YS-M2}, which is designed to take an $N$ particle state in the input and separate out a state with only two particles in the output. It splits each input mode 1 \& 2 into \{1,1''\} and \{2,2''\} respectively. Then, the number of particles is measured in the auxiliary modes, and only the cases with $s$ particles in mode 1'' and $N-s-2$ particles in mode 2'' are retained. Such a defined post-selection, for each $s=0,1,...\,,N-2$, warrants that two particles in modes 1 \& 2 remain. This filtering prepares some \textit{event-ready} state on which the Yurke-Stoler to be performed. If no non-local correlations are to be observed with such a prepared state, then certain conditions must be met as explained in \textbf{Lemma~\ref{lemma}}. Calculating explicitly, we get the following set of constraints
\begin{eqnarray}\label{Condition-YS-M2}
\beta_{s+1}^2&=&\beta_{s}\cdot\beta_{s+2}\cdot\sqrt{\tfrac{s+2}{s+1}}\,{\sqrt{\tfrac{N-s}{N-s-1}}}\,,
\end{eqnarray}
for each $s=0,1,...\,,N-2$.

As it turns out, those conditions suffice to analyse all states except for the so-called NOON states, which have the form
\begin{eqnarray}\label{psi-M2-NOON}
\ket{\phi}_{\!\scriptscriptstyle{NOON}}&=&\Big(\,\beta_0\,\tfrac{{a_2^\dag}^{\scriptscriptstyle{N}}}{\sqrt{N!}}+\beta_N\,\tfrac{{a_1^\dag}^{\scriptscriptstyle{N}}}{\sqrt{N!}}\,\Big)\ket{0}\,,
\end{eqnarray}
i.e., when $\beta_1=...=\beta_{N-1}=0$ in Eq.~(\ref{psi-M2}). In this case, we need to modify the optical setup by adding \textit{quantum erasure}, as shown in Fig.~\ref{Fig-YS-QE-M2}. Here, the information about the number 
of particles in the respective modes 1'' and 2'' gets erased when the Hadamard transformation precedes the detection. This is a straightforward generalisation of the erasure of the which-path information for a single photon~\cite{ScEnWa91}. For our purpose, it is sufficient to restrict the attention to the post-selected system with $N-2$ particles in mode 1'' and $0$ particles in mode 2''. Then the resulting \textit{event-ready} state in modes 1 and 2 is subject to the Yurke-Stoler test. Explicit calculation shows that the procedure does not reveal non-local correlations only if $\beta_0=0$ or $\beta_N=0$, which follows from \textbf{Lemma~\ref{lemma}}. That is, we have
\begin{eqnarray}\label{Condition-YS-M2-NOON}
\!\!\!\!\beta_1=...=\beta_{N-1}=0&\ \Rightarrow\ &\beta_0=0\ \ \text{or}\ \ \beta_N=0\,, 
\end{eqnarray}
meaning that every proper NOON state (for $\beta_0,\beta_N\neq0$) is Bell non-local.

We have thus obtained enough constraints to provide the following technical result. 
\begin{lemma}\label{psi-M2-lemma}\ \\
The resolution of conditions in Eq.~(\ref{Condition-YS-M2}) and (\ref{Condition-YS-M2-NOON}) imply that $\beta_n=\big(\!\begin{smallmatrix}N\\n\end{smallmatrix}\!\big)^{\!\nicefrac{1}{2}}\,U_1^n\,U_2^{N-n}$ for some $|U_1|^2+|U_2|^2=1$.
\end{lemma}
By comparing with Eq.~(\ref{single-mode-coeff-product}), we conclude that any state $\ket{\phi}$ in Eq.~(\ref{Psi-M2-N2}) which does \myuline{not} show non-local correlations in the experimental setups in Fig.~\ref{Fig-YS-M2} and \ref{Fig-YS-QE-M2} is necessarily of a single-mode type. This proves \textbf{Claim~\ref{claim1}} for the case $N\!\geqslant\!2$ and $M\!=\!2$.
\begin{figure}[t]
\centering
\includegraphics[width=\columnwidth]{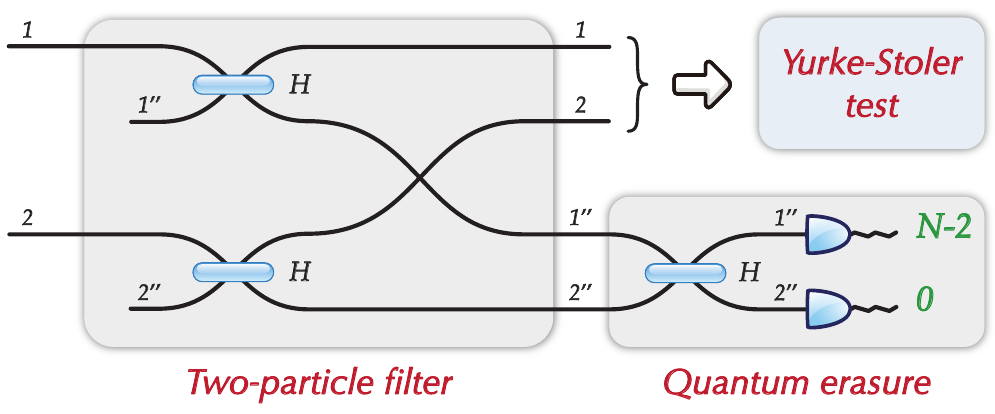}
\caption{\label{Fig-YS-QE-M2}{\bf\textsf{\mbox{Yurke-Stoler type experiment with quantum erasure.}}}\\ Modification of the setup in Fig.~\ref{Fig-YS-M2} which erases the information about the number of particles in modes 1'' \& 2''.\vspace{-0.2cm}}
\end{figure}
\vspace{0.2cm}

\textsf{\textit{$\bullet$ \myuline{Many particles \& many modes ($N\!\geqslant\!2, M\!\geqslant\!2$)}}}\vspace{0.1cm}

This is the most general case covered by \textbf{Claim~\ref{claim2}}. To prove this result, we proceed by induction on the number of modes $M\!=\!2,3,...$ (for any number of particles $N\!\geqslant\!2$). That is, we consider the following series of statements:\vspace{0.2cm}

\textit{1.} \textbf{Claim~\ref{claim2}} holds for $M\!=\!2$ (and any $N\!\geqslant\!2$).

\textit{2.} \textbf{Claim~\ref{claim2}} holds for $M\!=\!3$ (and any $N\!\geqslant\!2$).

\textit{3.} \textbf{Claim~\ref{claim2}} holds for $M\!=\!4$ (and any $N\!\geqslant\!2$).

...\quad...\quad...\quad...\quad...\quad...\quad...\quad...\quad...\quad...\vspace{0.2cm}

By mathematical induction, it is sufficient to justify:\vspace{-0.1cm}

\begin{itemize}
\item[\textit{a)}]{\textit{Base case ($M\!=\!2$):}\\\textbf{Claim~\ref{claim2}} holds for $M\!=\!2$ (and any $N\!\geqslant\!2$).\vspace{-0.1cm}}
\item[\textit{b)}]{\textit{Inductive step ($M-1\!\leadsto\!M$):}\\If claim \textbf{Claim~\ref{claim2}} holds for $M-1$ (and any $N\!\geqslant\!2$), then it also holds for $M$ (and any $N\!\geqslant\!2$), for $M\geqslant3$.\vspace{-0.1cm}}
\end{itemize}

Note that the \textit{base case} was proved above; see \textsf{\textit{Many particles \& two modes}} ($N\!\geqslant\!2, M\!=\!2$). For the proof of the inductive step, we refer to \textbf{Supplementary Information}.
\vspace{0.1cm}

\myuline{\textbf{\textit{Sketch of proof of \textbf{Claim~\ref{claim2}}}}}\ \vspace{0.2cm}\\\indent
We aim to construct a generic local hidden variable model that is capable of simulating any passive linear optical experiment with the single-mode type states. This will require defining the \textit{ontology} of what propagates in the circuit and how the \textit{evolution} is implemented by the gates. Crucially, the propagation of the variables in the model must follow the local pattern of interconnections in a given experimental design.

Without loss of generality, we may consider two-mode gates ({beam splitters}) and single-mode gates ({phase shifters} and {detectors}). These basic elements wired into a complex circuit construct any passive linear optical experiment~\cite{ReZeBeBe94}.
Let us briefly recall how those ecomponents are described in quantum theory. A {beam splitter} is associated a the 2x2 unitary transformation $\mathbb{V}$ which acts on the intersecting paths $s$ and $t$ as follows
\begin{eqnarray}\label{BS}
\binom{a^\dag_s}{a^\dag_t}&\xymatrix{\ar[r]^{\atop } &}& \binom{{a'\!}_s^{\,\,\dag}}{{a'\!}_t^{\,\,\dag}}\,=\,\mathbb{V}\,\binom{a^\dag_s}{a^\dag_t}\,.
\end{eqnarray}
A {phase shifter} in path $s$ introduces a phase $e^{i\varphi}$ where it is located
\begin{eqnarray}\label{PS}
a^\dag_s&\xymatrix{\ar[r]^{\atop } &}& {a'\!}_s^{\,\,\dag}\,=\,e^{i\varphi}\,a^\dag_s\,.
\end{eqnarray}
Finally, there are single-mode detectors which terminate evolution in a given path and provide classical information about the number of particles in that path. In our discussion of passive linear optical setups, their action can be effectively pulled out to the end of the experiment since all paths are eventually measured. Accordingly, an array of detectors that measures the state in Eq.~(\ref{N-particle-state-psi}) will register $n_1,n_2,...\,,n_M$ particles in the respective modes with the following statistics
\begin{eqnarray}\label{quantum-statistics}
\ket{\psi}&\Longrightarrow&\text{Prob}_{\ket{\psi}}(n_1,...\,,n_M)\,=\,|\psi_{n_1\,...\,n_M}|^2\,.
\end{eqnarray}
This completes the list of basic elements of passive linear optics and their quantum-mechanical description. Notably, for a general multi-particle state, those rules lead to non-local correlations. In the following, we construct a generic local hidden variable model for the restricted class of single-mode type states.
\vspace{0.2cm}

\textsf{\textit{$\bullet$ \myuline{Construction of a generic local model}}}\vspace{0.1cm}

We will crucially exploit the property of single-mode type states, which allow for a compact representation
\begin{eqnarray}\label{psi-alpha}
\ket{\psi}&\stackrel{\text{1:1}}{\xlrsquigarrow}&\vec{\alpha}_\psi\, .
\end{eqnarray}
It is defined through Eq.~(\ref{single-mode-coeff-product}) where $\vec{\alpha}_\psi\equiv(U_1,...\,,U_M)$. We note that this property is specific to the single-mode type states since, in general, the coefficients $\psi_{n_1\,...\,n_M}$ in Eq.~(\ref{N-particle-state-psi}) cannot be compactly encoded in only $M$ numbers.
\begin{lemma}\label{lemma-psi-alpha}\ \\
For each single-mode type state $\ket{\psi}$, the associated vector $\vec{\alpha}_\psi$ is \myuline{uniquely} defined (up to the overall phase). 
\end{lemma}

Let us define the \textit{ontology} of the model by postulating that each mode $i=1,...\,,M$ carries a complex number $\alpha_i$  and an integer $k_i$. The $\alpha_i$'s will be called \textit{amplitudes}, and $k_i$'s will be identified with the \textit{number of particles} in $i$-th mode $i$. The \textit{hidden variable space} is thus defined as 
\begin{eqnarray}\label{HV-space}
\Lambda&:=&(\mathbb{C}\times\mathbb{N})^M\ =\ \mathbb{C}^M\times\mathbb{N}^M\ni(\vec{\alpha},\vec{k})\,.
\end{eqnarray}
It is composed of subspaces $\Lambda_i=\mathbb{C}\times\mathbb{N}$ associated with each individual local mode, i.e. $\Lambda=\Lambda_1\times...\times\Lambda_M$. In the following we will use the notation ${\lambda}\equiv(\vec{\alpha},\vec{k})\!\in\!\Lambda$.

Furthermore, to complete the model, we need to define how the 
ontic states $(\vec{\alpha},\vec{k})\rightarrow (\vec{\alpha}',\vec{k}')$ \textit{evolve} under the action of the basic elements building passive linear optical setups. Importantly, the \textit{locality} condition requires that the gates defined below involve only the variables propagating along the paths meeting in the gate.

For the \textit{beam splitter} associated with the unitary $\mathbb{V}$ in paths $s$ and $t$, we postulate
\begin{eqnarray}\label{def-beam-splitter}
\!\!\begin{array}{lll}
(\alpha_s,\alpha_t)&\!\!\longrightarrow\!\!&(\alpha_s',\alpha_t')\,:=\,(\alpha_s,\alpha_t)\,\mathbb{V}\,,\vspace{0.1cm}\\
(k_s,k_t)&\!\!\longrightarrow\!\!&(k_s',k_t')\ \ \textit{with prob.}\ \begin{pmatrix}k\\k_s'\,,k_t'\end{pmatrix}\tfrac{|\alpha_s'|^{2k_s'}\,|\alpha_t'|^{2k_t'}}{(|\alpha_s'|^2+|\alpha_t'|^2)^k}\,,
\end{array}\ \ 
\end{eqnarray}
where $k:=k_s+k_t=k_s'+k_t'=:k'$ is the number of particles, which is preserved. After the gate, the particles are multinomially distributed in modes $s$ and $t$ with probabilities $p_s=\nicefrac{|\alpha_s'|^{2}}{(|\alpha_s'|^2+|\alpha_t'|^2)}$ and $p_t=\nicefrac{|\alpha_t'|^{2}}{(|\alpha_s'|^2+|\alpha_t'|^2)}$. Note that the latter are well-normalised and given by the transformed amplitudes $(\alpha_s',\alpha_t')$. Such a defined beam splitter is a local stochastic gate that operates exclusively in the intersecting paths.

The \textit{phase shifter} in path $s$ introduces phase $e^{i\varphi}$ to the amplitude in that path, i.e, we have
\begin{eqnarray}\label{def-phase-shifter}
\alpha_s\ \longrightarrow\ \alpha_s'\,:=\,e^{i\varphi}\alpha_s &\ \text{and}\ &
k_s\ \longrightarrow\ k_s'\,:=\,k_s\,,
\end{eqnarray}
which is a local deterministic gate.

Finally, each path ends with the \textit{number-resolving detector}, which reveals the number of particles\begin{eqnarray}\label{def-detector}
(\alpha_s,k_s)\ \Longrightarrow\ k_s\,.
\end{eqnarray}
Its action is action clearly local.

This completes the list of basic elements in the model that enable the construction of any passive optical setup. It remains to demonstrate that the model reconstructs all the quantum-mechanical predictions for single-mode type states in any passive linear optical experiment. 
\vspace{0.2cm}

\textsf{\textit{$\bullet$ \myuline{Check of quantum-like behaviour}}}\vspace{0.1cm}

Let us postulate the \textit{epistemic state} $\mu_\psi({\lambda})$ associated with the single-mode type quantum state $\ket{\psi}$. Using the representation in Eq.~(\ref{psi-alpha}), we define
\begin{eqnarray}\label{epistemic-state}
\mu_\psi({\lambda})\,\equiv\,\mu_\psi(\vec{\alpha},\vec{k})
\,:=\,\tfrac{1}{2\pi}\,\delta_{\vec{\alpha}\sim\vec{\alpha}_\psi}\,\binom{N}{\vec{k}}\,\prod_{i=1}^M\,|{\alpha}_{\scriptscriptstyle{i}}|^{2k_i}\,,\ \ \ \ 
\end{eqnarray}
where $\delta_{\vec{\alpha}\sim\vec{\alpha}_\psi}$ is the Dirac $\delta$-function equal $1$ if $\vec{\alpha}=e^{i\varphi}\,\vec{\alpha}_\psi$ for some $\varphi\in[0,2\pi)$, otherwise equal $0$. In other words, the epistemic state $\mu_\psi({\lambda})$ is \textit{uniformly} distributed on $\vec{\alpha}$'s equal up to overall phase to $\vec{\alpha}_{\psi}$, and $\vec{k}$'s are \textit{multinomially} distributed with probabilities $p_i=|{\alpha_\psi}_{\scriptscriptstyle{i}}|^2$. 

The crucial observation is that the above-defined class of epistemic states
$\mu_\psi({\lambda})$ is closed under the dynamics generated by the gates in the model and the evolution is congruent with the quantum-mechanical description.

\begin{lemma}[Equivariance]\label{equivariance}\ \\
The epistemic states in Eq.~(\ref{epistemic-state}) are preserved under the action of any configuration of gates defined in the model.
Furthermore, their evolution $\mu_\psi({\lambda})\rightarrow\mu_{\psi'}({\lambda})$ follows the quantum-mechanical rules $\psi\rightarrow{\psi'}$ as laid out in Eqs.~(\ref{BS}) and (\ref{PS}), and the observed experimental statistics is given by Eq.~(\ref{quantum-statistics}).
\end{lemma}

In conclusion, when restricted to the class of epistemic states $\mu_\psi({\lambda})$, the predictions of the model are identical to the quantum-mechanical description of single-mode type states. This is sufficient to justify the existence of a generic local hidden variable model that simulates the behaviour of any single-mode type state in arbitrary passive linear optical experiments.
Notice the resemblance of our construction model to the de~Broglie-Bohm (or pilot wave) theory~\cite{Bo52}, with the caveat that the latter is non-local. Our model can be viewed as the maximum achievable as regards local simulation of multi-particle quantum states in passive linear optical setups.

\vspace{0.5cm}
\noindent{\small\bf\textsf{{Acknowledgements}}}\\
We thank Andrew Jordan, Matthew Leifer and Marek Żukowski for helpful discussions.
PB acknowledges support from the Polish-U.S. Fulbright Commission.
MM acknowledges partial support from the Foundation for Polish Science (IRAP project, ICTQT, contract no. MAB/2018/5, co-financed by EU within the Smart Growth Operational Programme).


\vspace{0.5cm}
\noindent{\small\bf\textsf{Additional information}}\\
The authors declare no competing interests.

\bibliography{/Users/pblasiak/GoogleDrive/Library/CombQuant}

\newpage\ \newpage

\title[]{All non-local states of identical particles}


\author{Pawel \surname{Blasiak}}
\email{pawel.blasiak@ifj.edu.pl}
\affiliation{Institute for Quantum Studies, Chapman University, Orange, CA 92866, USA}
\affiliation{Institute of Nuclear Physics, Polish Academy of Sciences, 31342 Krak\'ow, Poland}

\author{Marcin \surname{Markiewicz}}
\email{marcinm495@gmail.com}
\affiliation{International Centre for Theory of Quantum Technologies, University of Gda\'nsk, 80308 Gda\'nsk, Poland}


\setcounter{page}{1} 

\onecolumngrid
\begin{center}
{\large\bf\text{\hypertarget{Supplementary-Information}{\myuline{\textit{Supplementary Information}}}}\vspace{0.3cm}\\Identical particles as a genuine non-local resource}\vspace{0.4cm}\\
Pawel Blasiak$^{\text{1,2}}$ and Marcin Markiewicz $^{\text{3,4}}$ \vspace{0.1cm}\\

\small\emph{$^{\text{1}}$Institute for Quantum Studies, Chapman University, Orange, CA 92866, USA\\
$^{\text{2}}$Institute of Nuclear Physics, Polish Academy of Sciences, 31342 Krak\'ow, Poland\\
$^{\text{3}}$Institute of Theoretical and Applied Informatics, Polish Academy of Sciences, 44100 Gliwice, Poland\\
$^{\text{4}}$International Centre for Theory of Quantum Technologies, University of Gda\'nsk, 80308 Gda\'nsk, Poland}
\end{center}
\vspace{0.4cm}

\noindent In this Supplement we complete the proof of \textbf{Theorem~\ref{theorem}} in the \textbf{Methods} section from the main manuscript.

\noindent \textit{[The numbering of equations follows the main text.]}

\vspace{0.4cm}
\begin{center}
\textsf{\textbf{A. Single-mode type states: Some useful facts and proof of Lemma~\ref{lemma-psi-alpha}}}\vspace{0.1cm}
\end{center}

By \textbf{Definition~\ref{Def-single-mode-type}}, for each single-mode type state, there is some unitary $\mathbb{U}\equiv [U_{ij}]$ that transforms it to a state with all particles in the same mode. A straightforward consequence is that
\begin{observation}
The class of \textbf{single-mode type states} is \myuline{invariant} with respect to passive linear optical transformations.
\end{observation}

Furthermore, using Eq.~(\ref{single-mode-def}), we can write
\begin{eqnarray}
\ket{\psi}&=&\mathbb{U}\ket{N,0,...\,,0}\,=\,\frac{{{a'\!}_1^{\,\,\dag}}^{N}}{\sqrt{N!}}\ket{0}\ \stackrel{\scriptscriptstyle{(\ref{U-a-dag})}}{=}\ \frac{1}{\sqrt{N!}}\Big[\,{\sum}_{i=1}^M\,U_{1i}\,a_i^\dag\,\Big]^N\ket{0}\ =\ \sum_{{|\vec{n}|=N}}\binom{N}{\vec{n}}^{\!\nicefrac{1}{2}}\,\prod_{i=1}^M\,U_{1i}^{n_i}\,\frac{{a_i^\dag}^{n_i}}{\sqrt{n_i!}}\ket{0},
\end{eqnarray}
Now, comparing the terms with Eq.~(\ref{N-particle-state-adag}) and taking $U_{i}\equiv U_{1i}$, we get Eq.~(\ref{single-mode-coeff-product}), i.e.,
\begin{eqnarray}\label{single-mode-coeff-product-lemma}
\psi_{n_1\,...\,n_M}&=&\binom{N}{\vec{n}}^{\!\nicefrac{1}{2}}\,\prod_{i=1}^M\,U_i^{n_i}.
\end{eqnarray}
This leads to the simple vector representation of single-mode type states in Eq.~(\ref{psi-alpha}), i.e.,
\begin{eqnarray}\label{psi-alpha-lemma}
\ket{\psi}&\stackrel{\text{1:1}}{\xlrsquigarrow}&\vec{\alpha}_\psi\,,
\end{eqnarray}
where $\vec{\alpha}_\psi\equiv(U_1,...\,,U_M)$ is uniquely defined (up to overall phase) as stated in \textbf{Lemma~\ref{lemma-psi-alpha}}.

\vspace{0.4cm}
\textbf{\textsf{\textit{A.1. {Proof of \textbf{Lemma~\ref{lemma-psi-alpha}}.}}}}\vspace{0.3cm}

Suppose we have two representations of the single-mode type state $\ket{\psi}$ via Eq.~(\ref{single-mode-coeff-product-lemma}), say by vectors $\vec{\alpha}$ and $\vec{\beta}$. This means that we have
\begin{eqnarray}\label{uniqueness-1}
\alpha_1^{n_1}\,\alpha_2^{n_2}...\,\alpha_M^{n_M}&=&e^{i\theta}\,\beta_1^{n_1}\,\beta_2^{n_2}...\,\beta_M^{n_M}\qquad\qquad\text{for all $n_1+n_2+...+n_M=N$},
\end{eqnarray}
where $e^{i\theta}$ is some fixed phase due to overall phase redundancy for quantum states. If we take $n_i=N$ (and the rest $0$), then it entails $|\alpha_i|=|\beta_i|$. Let us define $\gamma_i:=\tfrac{\alpha_i}{\beta_i}=e^{i\varphi_i}$. Then, the condition Eq.~(\ref{uniqueness-1}) rewrites as
\begin{eqnarray}\label{uniqueness-2}
\gamma_1^{n_1}\,\gamma_2^{n_2}...\,\gamma_M^{n_M}&=&e^{i\theta}\qquad\qquad\text{for all $n_1+n_2+...+n_M=N$}.
\end{eqnarray}
Taking again $n_i=N$ (and the rest $0$) we infer that $\gamma_i^N=e^{iN\varphi_i}=e^{i\theta}$, which gives
\begin{eqnarray}\label{uniqueness-3}
\varphi_i=\tfrac{\theta}{N}+k_i\tfrac{2\pi}{N}\qquad\text{for some $k_i=0,...,N-1$},
\end{eqnarray}
and the condition Eq.~(\ref{uniqueness-2}) rewrites as
\begin{eqnarray}\label{uniqueness-4}
n_1k_1+n_2k_2+...+n_Mk_M&=&Nl_{\vec{n}}\qquad\qquad\text{for all $n_1+n_2+...+n_M=N$},
\end{eqnarray}
for some $l_{\vec{n}}$'s in $\mathbb{N}$ (i.e., for each choice of $n_i$'s the expression on the left must be a multiple of $N$). 

Now, take any two indices $i$ and $j$ for which we put $n_i=N-1$ and $n_j=1$. Then the condition Eq.~(\ref{uniqueness-4}) gives
\begin{eqnarray}\label{uniqueness-5}
(N-1)k_i+k_j\ =\ Nl\quad\Rightarrow\quad k_j-k_i\ =\ N(l-k_i),\qquad\text{for some $l\in\mathbb{N}$}.
\end{eqnarray}
Since $|k_j-k_i|\leqslant N-1$ from Eq.~(\ref{uniqueness-3}), this necessitates that $k_j-k_i=0$. Therefore, we must have $\varphi_i=\varphi_j$  for all $i$ and $j$. This proves that both representations may differ only by an overall phase $\vec{\alpha}=e^{i\theta}\vec{\beta}$.

\vspace{0.4cm}
\textbf{\textsf{\textit{A.2. {Some useful facts about the representation in Eq.~(\ref{psi-alpha})/(\ref{psi-alpha-lemma}).}}}}\vspace{0.3cm}

One can easily calculate how the vector representation of single-mode type states $\vec{\alpha}_{\psi}$ transforms under the action of a unitary $\mathbb{U}=[U_{ij}]$ in Eq.~(\ref{U-a-dag}).
\begin{proposition}\label{alpha-transformation}
For a \textbf{single-mode type} state, its vector representation Eq.~(\ref{psi-alpha})/(\ref{psi-alpha-lemma}) transforms by multiplication on the \myuline{right}, i.e. we have
\begin{eqnarray}
\text{if}\ \ 
\ket{\psi}\ \xymatrix{\ar[r]^{\atop \mathbb{U}} &}\  \ket{\psi'}\ ,\ \ \text{then}\ \ 
\vec{\alpha}_{\psi}\ \xymatrix{\ar[r]^{\atop \mathbb{U}} &}\  \vec{\alpha}_{\psi'}\ =\ \vec{\alpha}_{\psi}\,\mathbb{U}\,.
\end{eqnarray}
\end{proposition}
\begin{proof}
From the \textbf{Definition~\ref{Def-single-mode-type}} of single-mode type states, we can write
\begin{eqnarray}\nonumber
\ket{\psi}=\tfrac{1}{\sqrt{N!}}\bigg[\sum_{i=1}^M{\alpha_\psi}_{\scriptscriptstyle{i}}\,a_i^\dag\bigg]^N\!\!\ket{0}&\xymatrix{\ar[r]^{\atop \mathbb{U}} &}&\tfrac{1}{\sqrt{N!}}\bigg[\sum_{i=1}^M{\alpha_\psi}_{\scriptscriptstyle{i}}\Big(\sum_{j=1}^MU_{ij}\,a_j^\dag\Big)\bigg]^N\!\!\ket{0}\ =\ \tfrac{1}{\sqrt{N!}}\bigg[\sum_{j=1}^M\Big(\sum_{i=1}^M{\alpha_\psi}_{\scriptscriptstyle{i}}U_{ij}\Big)a_j^\dag\bigg]^N\!\!\ket{0}\\
&&=\ \tfrac{1}{\sqrt{N!}}\bigg[\sum_{j=1}^M{\alpha_{\psi'}}_{\!\!\scriptscriptstyle{j}}\,a_j^\dag\bigg]^N\!\!\ket{0}\ =\ \ket{\psi'}.
\end{eqnarray}
where ${\alpha_{\psi'}}_{\!\!\scriptscriptstyle{j}}:=\sum_{i=1}^M{\alpha_\psi}_{\scriptscriptstyle{i}}U_{ij}$. This proves that $\vec{\alpha}_{\psi}$'s transform by multiplying on the right $\vec{\alpha}_{\psi'}\,=\,\vec{\alpha}_{\psi}\,\mathbb{U}$.
\end{proof}

For future reference, let us write down the quantum-mechanical distribution of particles in Eq.~(\ref{quantum-statistics}) using the vector representation $\vec{\alpha}_{\psi}$ in  Eq.~(\ref{psi-alpha})/(\ref{psi-alpha-lemma}).
\begin{observation}\label{observation-probability-single-mode-alpha}
A single-mode type state measured in all modes has the following statistics
\begin{eqnarray}\label{probability-single-mode-alpha}
\textnormal{Prob}_{\ket{\psi}}(n_1,...\,,n_M)&\stackrel{\scriptscriptstyle{(\ref{quantum-statistics})}}{=}&|\psi_{n_1\,...\,n_M}|^2\ \stackrel{\scriptscriptstyle{(\ref{single-mode-coeff-product-lemma})}}{=}\ \binom{N}{\vec{n}}\,\prod_{i=1}^M\,|{\alpha_\psi}_{\scriptscriptstyle{i}}|^{2n_i}\,.
\end{eqnarray}
\end{observation}

\vspace{0.4cm}
\textbf{\textsf{\textit{A.3. {Which states are of a single-mode type?}}}}\vspace{0.3cm}

For both bosons and fermions, all single-particle states ($N\!=\!1$) are of a single-mode type; see Ref.~\cite{ReZeBeBe94}. Clearly, in the multi-particle case ($N\!\geqslant\!2$), for fermions this class is empty since the Pauli exclusion principle rules out more than one particle in the same mode. The boson statistics admits a wider variety of states with the non-trivial class of single-mode type states in the multi-particle case ($N\!\geqslant\!2$); see Fig.~\ref{Fig-Clasiffication}.

The vector representation in Eq.~(\ref{psi-alpha})/(\ref{psi-alpha-lemma}) provides a neat characterisation of single-mode type states. For a state in the form Eq.~(\ref{N-particle-state-psi})/(\ref{N-particle-state-adag}), it is sufficient to verify if there exists a normalised vector $(U_1,...\,,U_M)$ such that Eq.~(\ref{single-mode-coeff-product})/(\ref{single-mode-coeff-product-lemma}) holds. This can be checked algorithmically in a few easy steps:\vspace{0.1cm}

\textit{(1) If it exists, then we need to have $U_i^N\,=\,\psi_{0...N...0}$ (with $N$ on $i$-th place, i.e. $n_i=N$).}\vspace{0.05cm}

\textit{(2) This means that $U_i$ is $N$th root of $\psi_{0...N...0}$. 
Calculate those roots.}\vspace{0.05cm}

\textit{(3) Then check whether there is a solution for which Eq.~(\ref{single-mode-coeff-product})/(\ref{single-mode-coeff-product-lemma}) holds for all coefficients $\psi_{n_1\,...\,n_M}$.
}\vspace{0.2cm}


Let us end with a few examples illustrating different types of states:
\begin{eqnarray}\nonumber
\begin{array}{llcll}
\textit{\myuline{Single-mode type states}:}\hspace{-0.5cm}&&
\ket{N}&&\textit{all particles in the same mode\,,}\vspace{0.1cm}\\
&&\alpha\ket{10}+\beta\ket{01}&&\textit{single-particle states ($N\!=\!1$)\,,}\vspace{0.1cm}\\
&&\tfrac{1}{2}(\ket{20}+\sqrt{2}\ket{11}+\ket{02})&\quad&\textit{two-particle state $\ket{20}$ after Hadamard gate (beam splitter)\,.}\vspace{0.1cm}\\\\
\textit{\myuline{Non single-mode type states}:}\hspace{-0.5cm}&&
\ket{k\,l}&&\textit{particles in different modes for $k,l\geqslant1$\,,}\vspace{0.1cm}\\
&&\alpha\ket{N0}+\beta\ket{0N}&&\textit{NOON-type states (for $N\!\geqslant\!2$)\,,}\vspace{0.1cm}\\
&&\tfrac{1}{\sqrt{2}}(\ket{1010}+\ket{0101})&&\textit{Bell state\,.}\vspace{0.1cm}\\
\end{array}
\end{eqnarray}

\clearpage

\vspace{0.4cm}
\begin{center}
\textsf{\textbf{B. Explicit calculations of conditions in Eqs.~(\ref{Condition-YS-M2}) and (\ref{Condition-YS-M2-NOON}), and proof of Lemma~\ref{psi-M2-lemma}}}\vspace{0.3cm}
\end{center}

In the \textbf{Methods} section, we considered two cases of passive linear optical experiments, as shown in Fig.~\ref{Figs-M2-proof}, designed for the purpose of analysing two categories of states. The first one (on the left) concerns a general multi-particle two-mode state in the form
\begin{eqnarray}\label{psi-M2-proof}
\ket{\phi}&=&\sum_{n=0}^N\beta_n\,\tfrac{{a_1^\dag}^n}{\sqrt{n!}}\tfrac{{a_2^\dag}^{N-n}}{\sqrt{(N-n)!}}\ket{0}\,,
\end{eqnarray}
while the second one (on the right) is applied to the special class of NOON states (for which $\beta_1=...=\beta_{N-1}=0$)
\begin{eqnarray}\label{psi-M2-NOON-proof}
\ket{\phi}_{\!\scriptscriptstyle{NOON}}&=&\Big(\,\beta_0\,\tfrac{{a_2^\dag}^{\scriptscriptstyle{N}}}{\sqrt{N!}}+\beta_N\,\tfrac{{a_1^\dag}^{\scriptscriptstyle{N}}}{\sqrt{N!}}\,\Big)\ket{0}\,.
\end{eqnarray}

In the following, we calculate the evolution of those states in the respective setups and draw conclusions from the respective Yurke-Stoler test.\vspace{0.2cm} 

\begin{figure*}[h]
\centering
\includegraphics[width=0.47\columnwidth]{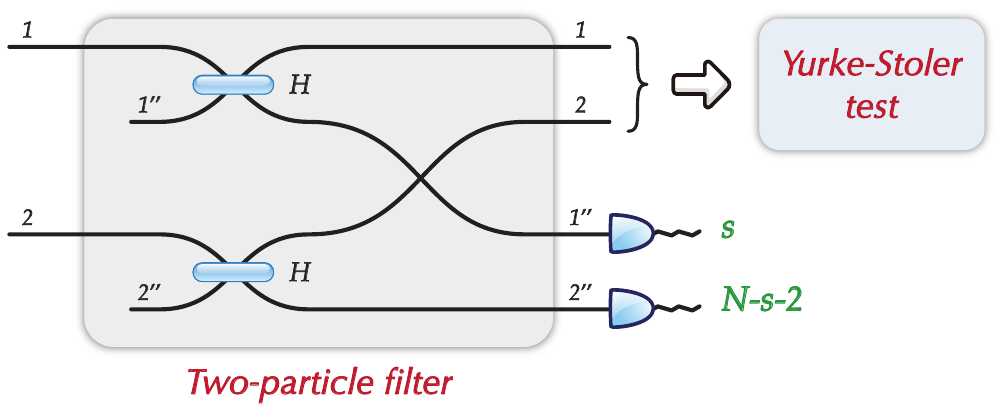}
\quad\qquad
\includegraphics[width=0.47\columnwidth]{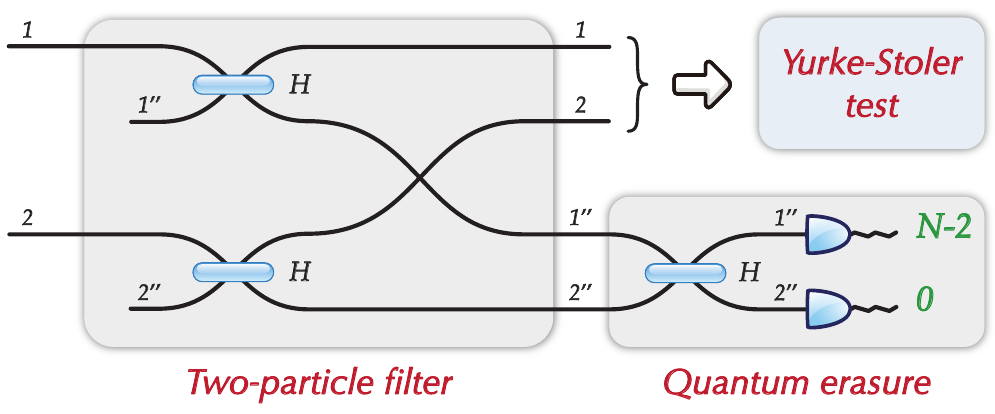}
\caption{\label{Figs-M2-proof}{\bf\textsf{\mbox{Yurke-Stoler type experiment with two-particle filter and its extension with quantum erasure (as in Figs.~\ref{Fig-YS-M2} and \ref{Fig-YS-QE-M2}).}}}\\
On the left, in the initial phase, $N-2$ particles are pulled out from the system by conditioning on $s$ and $N-s-2$ particle detections in modes 1'' \& 2''. For each $s=0,1,...\,,N-2$\,, this guarantees event-ready preparation of some two-particle state in modes 1 \& 2, which then undergo the Yurke-Stoler test. On the right,
modification of the setup on the left, designed to erase the information about the number of particles in modes 1'' \& 2''.
 }
\end{figure*}

\vspace{0.4cm}
\textbf{\textsf{\textit{B.1. {Conditions in Eq.~(\ref{Condition-YS-M2}) from the experiment in Fig.~\ref{Figs-M2-proof} (on the left).}}}}\vspace{0.3cm}

The state $\ket{\psi}$ in Eq.~(\ref{psi-M2-proof}) injected into paths 1 and 2 evolves through the two-particle filter in Fig.~\ref{Figs-M2-proof} (on the left), preparing an event-ready state heralded by the detection of $s$ and $N-s-2$ particles in paths 1'' \& 2''. Then, for each value of $s=0,1,...\,,N-2$, the post-selected system in paths 1 \& 2 is left in the following two-particle state
\begin{eqnarray}
\ket{\phi}&\xymatrix{\ar[r]^{\atop H,H}&}&\sum_{n=0}^N\ \beta_n\,\tfrac{1}{\sqrt{2}^{\,n}}\tfrac{\left(a_{1}^\dag+a_{1''}^\dag\right)^n}{\sqrt{n!}}\tfrac{1}{\sqrt{2}^{\,N-n}}\tfrac{{\left(a_{2}^\dag+a_{2''}^\dag\right)}^{N-n}}{\sqrt{(N-n)!}}\ket{0}\\\label{Two-particle-filter-M2-binomial-1}
&&=\ \tfrac{1}{\sqrt{2}^{\,N}}\sum_{n=0}^N\ \beta_n\,\tfrac{1}{\sqrt{n!\,(N-n)!}}\ \Big(\,{a_{1''}^\dag}^n+n\,{a_{1''}^\dag}^{n-1}\,{a_{1}^\dag}+\tfrac{n\,(n-1)}{2}\,{a_{1''}^\dag}^{n-2}\,{a_{1}^\dag}^{2}+...\,\Big)\\\label{Two-particle-filter-M2-binomial-2}
&&\qquad\qquad\qquad\cdot\,\Big(\,{a_{2''}^\dag}^{N-n}+({N-n})\,{a_{2''}^\dag}^{N-n-1}\,{a_{2}^\dag}+\tfrac{(N-n)\,(N-n-1)}{2}\,{a_{2''}^\dag}^{N-n-2}\,{a_{2}^\dag}^{2}+\,...\, \Big)\ket{0}
\\\label{Two-particle-filter-M2-1}
&\xrsquigarrow{\text{\!\tiny{\emph{post-select}}\!}}&\tfrac{1}{\sqrt{2}^{\,N}}\left(\ \beta_s\,\tfrac{1}{\sqrt{s!\,(N-s)!}}\tfrac{(N-s)\,(N-s-1)}{2}\ {a_{2}^\dag}^2\right.\\\label{Two-particle-filter-M2-2}
&&\qquad\quad\ \ +\ \beta_{s+1}\,\tfrac{1}{\sqrt{(s+1)!\,(N-s-1)!}}\,(s+1)\,(N-s-1)\ a_{1}^\dag\,a_{2}^\dag\\\label{Two-particle-filter-M2-3}
&&\qquad\quad\ \ \left.+\ \beta_{s+2}\,\tfrac{1}{\sqrt{(s+2)!\,(N-s-2)!}}\tfrac{(s+2)\,(s+1)}{2}\ {a_{1}^\dag}^2\ \right)\ket{0}.
\end{eqnarray}
In the above expressions, the brackets in Eqs.~(\ref{Two-particle-filter-M2-binomial-1}) and (\ref{Two-particle-filter-M2-binomial-2}) come from the binomial expansion. Then the post-selected terms in Eqs.~(\ref{Two-particle-filter-M2-1})\,-\,(\ref{Two-particle-filter-M2-3}) obtain by projection on the subspace corresponding to $s$ particles in mode 1'' and $N-s-2$ particles in mode 2'', i.e., we retain only the terms multiplying ${a_{1''}^\dag}^s\,{a_{2''}^\dag}^{N-s-2}$. For the sake of clarity, the resulting state is left unnormalised.

Now, upon successful post-selection in the auxiliary modes 1'' \& 2'', we perform the standard Yurke-Stoler test on the two-particle state in the output modes 1 and 2, given by Eqs.~(\ref{Two-particle-filter-M2-1})\,-\,(\ref{Two-particle-filter-M2-3}). Since this is a two-particle state, from \textbf{Lemma~\ref{lemma}}, we conclude that for this procedure \textit{not} to show Bell non-locality, the following condition has to be satisfied
\begin{eqnarray}
\Big(\,\beta_{s+1}\,\tfrac{(s+1)(N-s-1)}{\sqrt{(s+1)!(N-s-1)!}}\,\Big)^2&=&\beta_{s}\,\tfrac{(N-s)(N-s-1)}{\sqrt{s!(N-s)!}}\,\cdot\,\beta_{s+2}\,\tfrac{(s+2)(s+1)}{\sqrt{(s+2)!(N-s-2)!}}\,,
\end{eqnarray}
for each each instance of post-selection $s=0,1,...\,,N-2$. This is the condition $\beta^2=2\,\alpha\gamma$ for Eq.~(\ref{Psi-M2-N2}) form \textbf{Lemma~\ref{lemma}}, rewritten for the coefficients read out from Eqs.~(\ref{Two-particle-filter-M2-1})\,-\,(\ref{Two-particle-filter-M2-3}). After simplification, we get the constraints
\begin{eqnarray}\label{Condition-YS-M2-proof}
\beta_{s+1}^2\ =\ \beta_{s}\cdot\beta_{s+2}\cdot\sqrt{\tfrac{s+2}{s+1}}\,{\sqrt{\tfrac{N-s}{N-s-1}}}\,,\qquad\qquad\text{for\ \  $s=0,1,...\,,N-2$\,,}
\end{eqnarray}
which are the conditions in Eq.~(\ref{Condition-YS-M2}).

\vspace{0.4cm}
\textbf{\textsf{\textit{B.2. {Conditions in Eq.~(\ref{Condition-YS-M2-NOON}) from the experiment in Fig.~\ref{Figs-M2-proof} (on the right).}}}}\vspace{0.3cm}

For the analysis of the NOON state $\ket{\psi}_{\scriptscriptstyle{NOON}}$ in Eq.~(\ref{psi-M2-NOON-proof}), we augment the above setup with quantum erasure, as shown in Fig.~\ref{Figs-M2-proof} (on the right), and post-select on the detection of $N-2$ and $0$ particles in modes 1'' \& 2''. Then, the setup prepares the following event-ready state in modes 1 \& 2
\begin{eqnarray}
\ket{\phi}_{\scriptscriptstyle{NOON}}&\xymatrix{\ar[r]^{\atop H,H}&}&\Big(\,\beta_0\,\tfrac{1}{\sqrt{2}^{\,N}}\tfrac{\left(a_{1}^\dag+a_{1''}^\dag\right)^{\scriptscriptstyle{N}}}{\sqrt{N!}}+\beta_N\,\tfrac{1}{\sqrt{2}^{\,N}}\tfrac{\left(a_{2}^\dag+a_{2''}^\dag\right)^{\scriptscriptstyle{N}}}{\sqrt{N!}}\,\Big)\ket{0}
\\
&\xymatrix{\ar[r]^{\atop H}&}&\tfrac{1}{\sqrt{2}^{\,N}}\tfrac{1}{\sqrt{N!}}\ \Big(\,\alpha_0\,\big(\,a_{1}^\dag+\tfrac{1}{\sqrt{2}}(\,a_{1''}^\dag+a_{2''}^\dag)\big)^{\scriptscriptstyle{N}}+\beta_N\,\big(\,a_{2}^\dag+\tfrac{1}{\sqrt{2}}\,(a_{1''}^\dag-a_{2''}^\dag)\big)^{\scriptscriptstyle{N}}\,\Big)\ket{0}
\\\label{Two-particle-filter-M2-NOON}
&\xrsquigarrow{\text{\!\tiny{\emph{post-select}}\!}}&\tfrac{1}{\sqrt{2}^{\,N}}\tfrac{1}{\sqrt{N!}}\tfrac{1}{\sqrt{2}^{\,N-2}}\tfrac{N(N-1)}{2}\ \Big(\,\alpha_0\,{a_{1}^\dag}^2+\beta_N\,{a_{2}^\dag}^2\,\Big)\ket{0}.
\end{eqnarray}
As above, the last line describes the state after successful post-selection, i.e., it retains only the terms multiplying the expression ${a_{1''}^\dag}^{N-2}\,{a_{2''}^\dag}^{0}$ in the binomial expansion. This state is left unnormalised as well. 

Then, such a prepared system undergoes the Yurke-Stoler test in modes 1 \& 2. This procedure will \textit{not} reveal non-local correlations only if $\beta_0=0$ or $\beta_N=0$, which follows from \textbf{Lemma~\ref{lemma}} (since then we need to have $0=2\,\beta_0\,\beta_N$). In other words, using the notation of Eq.~(\ref{psi-M2-proof}), we get
\begin{eqnarray}\label{Condition-YS-M2-NOON-proof}
\beta_1=\beta_2=...=\beta_{N-1}=0&\quad\Longrightarrow\quad&\beta_0=0\ \ \text{or}\ \ \beta_N=0\,,
\end{eqnarray}
which is the condition in Eq.~(\ref{Condition-YS-M2-NOON}). It means that every proper NOON-type state (for $\beta_0,\beta_N\neq0$) is a genuine non-local resource.

\vspace{0.4cm}
\textbf{\textsf{\textit{B.3. {Proof of \textbf{Lemma~\ref{psi-M2-lemma}}.}}}}\vspace{0.3cm}

We need to show that the conditions in Eqs.~(\ref{Condition-YS-M2}) and (\ref{Condition-YS-M2-NOON}) (or Eqs.~(\ref{Condition-YS-M2-proof}) and (\ref{Condition-YS-M2-NOON-proof})) lead to the solution in the form 
\begin{eqnarray}\label{recurrence-resolution-proof}
\beta_n&=&\binom{N}{n}^{\!\nicefrac{1}{2}}\,U_1^n\,U_2^{N-n}\,,
\end{eqnarray}
for some choice of normalised complex parameters $U_1$ and $U_2$, i.e., $|U_1|^2+|U_2|^2=1$.

Let us look at the structure of the condition in Eq.~(\ref{Condition-YS-M2}) in which all the coefficients, except the two extreme ones $\beta_0$ and $\beta_N$, are determined by their two closest neighbours as shown below\vspace{0.2cm}
\begin{eqnarray}
\qquad\xymatrix{\beta_{\scriptscriptstyle{0}}\ar@/^/[r]&\beta_{\scriptscriptstyle{1}}\ar@/_/[r]&\ar@/_/[l]\beta_{\scriptscriptstyle{2}}\ar@/^/[r]&...\ar@/^/[l]&...&...\ar@/^/[r]&\ar@/^/[l]\beta_{\scriptscriptstyle{N-2}}\ar@/_/[r]&\ar@/_/[l]\beta_{\scriptscriptstyle{N-1}}&\beta_{\scriptscriptstyle{N}}\ar@/^/[l]}
\end{eqnarray}

\vspace{0.2cm}
\noindent with the arrows depicting functional dependencies. From the fact that each coefficient in the middle is given by a product of their two neighbours, see Eq.~(\ref{Condition-YS-M2}), we observe that 
\begin{observation}
If $\beta_i=0$ for some $0\leqslant i\leqslant N$, then we necessarily have $\beta_1=\beta_2=...=\beta_{N-1}=0$.
\end{observation}
Thus, we are left with the following three trivial possibilities:
\begin{itemize}
\item[\textit{(1)}] $\quad\beta_0=0\ \ \text{or}\ \ \beta_N=0\quad\&\quad\beta_1=\beta_2=...=\beta_{N-1}=0$,
\item[\textit{(2)}] $\quad\beta_0,\beta_N\neq0\quad\&\quad\beta_1=\beta_2=...=\beta_{N-1}=0$,
\item[\textit{(3)}] $\quad\beta_0,\beta_N\neq0\quad\&\quad\beta_1,\beta_2,...\,,\beta_{N-1}\neq0$.
\end{itemize}
We will separately analyse each case.

\begin{proof}[\myuline{\textit{Case (1)}}]\ \\
It is a single-mode type state by definition. This is because then we have $\beta_0=1$ or $\beta_N=1$, which corresponds to the choice $U_1=1$, $U_2=0$ or $U_1=0$, $U_2=1$ in Eq.~(\ref{recurrence-resolution-proof}).
\end{proof}

\begin{proof}[\myuline{\textit{Case (2)}}]\ \\
It is ruled by the condition in Eq.~(\ref{Condition-YS-M2-NOON}) for $N\geq3$. We note that it is the only place where the discussion of the NOON states, giving the condition Eq.~(\ref{Condition-YS-M2-NOON}), is required. For $N=2$ we get a contradiction with Eq.~(\ref{Condition-YS-M2}), since then the latter implies $0=\beta_1^2=2\,\beta_0\,\beta_2\neq0$.
\end{proof}

\begin{proof}[\myuline{\textit{Case (3)}}]\ \\
Since all coefficients $\beta_0,\beta_1,...\,,\beta_{N}\neq0$, then we can rewrite the condition in Eq.~(\ref{Condition-YS-M2}) in the form
\begin{eqnarray}\label{Case-3-recurrence}
\beta_{s+2}\ =\ \frac{\beta_{s+1}^2}{\beta_{s}}\cdot\sqrt{\tfrac{s+1}{s+2}}\,{\sqrt{\tfrac{N-s-1}{N-s}}}\,,\qquad\qquad\text{for\ \  $s=0,1,...\,,N-2$\,.}
\end{eqnarray}
Note that by assumption, we do not have division by zero (it is the reason why the \textit{Case (2)} is considered separately and where the condition Eq.~(\ref{Condition-YS-M2-NOON}) is used).
This is recurrence for $\beta_2,\beta_3,...\,,\beta_{N}$ which are  \textit{uniquely} determined by two initial conditions $\beta_0$ and $\beta_1$. We will show that Eq.~(\ref{recurrence-resolution-proof}) provides a general solution. It suffice to check that Eq.~(\ref{recurrence-resolution-proof}):

\textit{(a)} satisfies the recurrence in Eq.~(\ref{Case-3-recurrence}), and

\textit{(b)} does not restrict the choice of initial conditions $\beta_0,\beta_1\neq0$.
\begin{itemize}
\item[]{\textit{Ad. (a).}
Check of the recurrence:
\begin{eqnarray}
\qquad\text{R.H.S. of Eq.~(\ref{Case-3-recurrence})}&\stackrel{\scriptscriptstyle{(\ref{recurrence-resolution-proof}})}{=}&\binom{N}{s+1}\ U_1^{2\,(s+1)}\,U_2^{2\,(N-s-1)}\cdot\binom{N}{s}^{\!-\nicefrac{1}{2}}U_1^{-s}\,U_2^{-(N-s)}\cdot\sqrt{\tfrac{s+1}{s+2}}\,{\sqrt{\tfrac{N-s-1}{N-s}}}\\
&=&\tfrac{N!}{(s+1)!\,(N-s-1)!}\,\sqrt{\tfrac{s!\,(N-s)!}{N!}}\,\sqrt{\tfrac{s+1}{s+2}}\,{\sqrt{\tfrac{N-s-1}{N-s}}}\ U_1^{s+2}\,U_2^{N-s-2}\\
&=&\sqrt{\tfrac{N!}{(s+2)!\,(N-s-2)!}}\ U_1^{s+2}\,U_2^{N-s-2}\\
&\stackrel{\scriptscriptstyle{(\ref{recurrence-resolution-proof}})}{=}&\text{L.H.S. of Eq.~(\ref{Case-3-recurrence})}\,.
\end{eqnarray}
}
\item[]{\textit{Ad. (b).} Arbitrary initial conditions $\beta_0,\beta_1\neq0$ can be expressed with Eq.~(\ref{recurrence-resolution-proof}) by choosing
\begin{eqnarray}
\qquad U_2\ =\ \beta_0^{\,\nicefrac{1}{N}}&\quad\text{and}\quad& U_1\ =\ \tfrac{1}{\sqrt{N}}\,U_2^{1-N}\,\beta_1\ =\ \tfrac{1}{\sqrt{N}}\,\beta_0^{\,\nicefrac{1}{N}\,-1}\,\beta_1\,.
\end{eqnarray}
}
\end{itemize}
\end{proof}
Thus, we have shown that each \textit{Case~(1)-(3)} leads to the form of Eq.~(\ref{recurrence-resolution-proof}). This concludes the proof of \textbf{Lemma~\ref{psi-M2-lemma}}.

\clearpage

\vspace{0.4cm}
\begin{center}
\textsf{\textbf{C. Proof of \textbf{Claim~\ref{claim2}} in the general case  $\bm{(N\!\geqslant\!2, M\!\geqslant\!2)}$: The inductive step $\bm{(M-1\!\leadsto\!M)}$}}\vspace{0.3cm}
\end{center}

To prove the result, we use mathematical induction on the number of modes $M\!=\!2,3,4,...$ and without making any restrictions on the number of particles $N\geqslant2$.

\begin{itemize}
\item[\textit{a)}]{\textbf{\textit{\myuline{Base case ($\bm{M\!=\!2}$)}}:}\vspace{0.1cm}\ \\ \textbf{Claim~\ref{claim2}} holds for $M\!=\!2$ (and any $N\!\geqslant\!2$).}
\end{itemize}

In other words, we have that any two-mode state
\begin{eqnarray}\label{psi-M2-general-proof}
\ket{\phi}&=&\sum_{n_1+n_2=N}\psi_{n_1n_2}\,\tfrac{{a_1^\dag}^{n_1}}{\sqrt{n_1!}}\tfrac{{a_2^\dag}^{n_2}}{\sqrt{n_2!}}\ket{0},
\end{eqnarray}
which \textit{does not} show non-local correlations has to be of the single-mode type, i.e., can be written in the form
\begin{eqnarray}\label{psi-M2-single-mode-general-proof}
\ket{\phi}&\stackrel{\scriptscriptstyle{(\ref{single-mode-def})}}{=}&\mathbb{U}\ket{N,0}\ =\ \mathbb{U}\,\frac{{a_1^\dag}^N}{\sqrt{N!}}\ket{0}\ \stackrel{\scriptscriptstyle{(\ref{U-a-dag})}}{=}\ \frac{1}{\sqrt{N!}}\Big[\,U_1\,a_1^\dag+U_2\,a_2^\dag\,\Big]^N\ket{0}\,,
\end{eqnarray}
for some unitary $\mathbb{U}$ (i.e., we have $U_1,U_2$ such that $|U_1|^2+|U_2|^2=1$).

The \textit{base case} was proven in the \textbf{Methods} section, see \textsf{\textit{Many particles \& two modes}} ($N\!\geqslant\!2, M\!=\!2$).
It remains to justify the \textit{inductive step}.

\begin{itemize}
\item[\textsl{b)}]{\textbf{\textit{\myuline{Inductive step ($\bm{M-1\!\leadsto\!M}$)}}:}\vspace{0.1cm}\ \\If claim \textbf{Claim~\ref{claim2}} holds for $M-1$ (and any $N\!\geqslant\!2$), then it also holds for $M$ (and any $N\!\geqslant\!2$), for $M\geqslant3$.}
\end{itemize}

That is, assuming that the result holds for $M-1$, we seek to find  constraints on a general $M$-mode state\begin{eqnarray}\label{N-particle-state-adag-general-proof}
\ket{\psi}&=&\sum_{n_1+\,...\,+n_M=N}\psi_{n_1\,...\,n_M}\ \prod_{i=1}^M\,\frac{{a_i^\dag}^{n_i}}{\sqrt{n_i!}}\ket{0},
\end{eqnarray}
requiring that it \textit{does not} exhibit non-local correlations in any passive linear optical experiment. Those constraints derived from considering certain experimental setups should constrain the state to the single-mode type form
\begin{eqnarray}\label{WWW}
\ket{\psi}&\stackrel{\scriptscriptstyle{(\ref{single-mode-def})}}{=}&\mathbb{W}\ket{N,0,...\,,0}\ =\ \mathbb{W}\,\frac{{a_1^\dag}^N}{\sqrt{N!}}\ket{0}\ \stackrel{\scriptscriptstyle{(\ref{U-a-dag})}}{=}\ \frac{1}{\sqrt{N!}}\Big[\,\sum_{i=1}^N\ W_i\,a_i^\dag\,\Big]^N\ket{0}\,,
\end{eqnarray}
where $\mathbb{W}$ is some unitary transformation  (i.e., we have $W_1,W_2,...,W_M$ such that $\sum_{i=1}^M|W_i|^2=1$).

Now we provide the proof of the \textit{inductive step}, which consists of two parts as illustrated in Fig~\ref{Fig-InductiveStep}. 

\begin{figure}[h]
\centering
\includegraphics[width=\columnwidth]{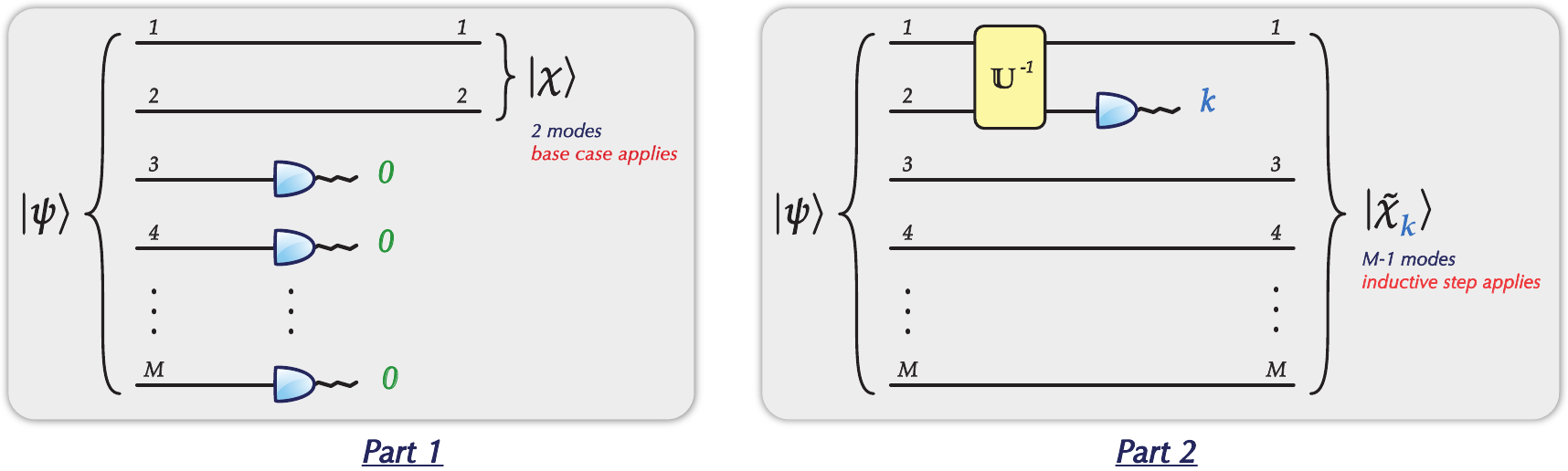}
\caption{\label{Fig-InductiveStep}{\bf\textsf{\mbox{Two parts in the proof of the inductive step.}}} On the left, conditioning on zero particle detections in modes 3,\,...\,,M prepares an event-ready two-mode state with $N$ particles in modes 1 \& 2. Because it is supposed not to show non-local correlations, it has to be of a single-mode type. From this observation, we get a useful unitary $\mathbb{U}$ simplifying the state $\ket{\psi}$. On the right, the protocol demonstrates that $\mathbb{U}^{-1}$ applied to modes 1 \& 2 of state $\ket{\psi}$ removes all the particles from mode 1. It is a conclusion from a series of constraints obtained from the inductive assumption applied to the event-ready state prepared in $M-1$ modes 1,3,\,...\,,M by conditioning on the detection of $k$ particles in mode 2.}
\end{figure}

\vspace{0.4cm}
\textsf{\textit{\textbf{\myuline{Part 1}.} {Simplification (from the base case $M\!=\!2$)}}.}\vspace{0.3cm}

Let us rewrite Eq.~(\ref{N-particle-state-adag-general-proof}) by separating out the terms with \myuline{\textit{all}} particles in modes 1 \& 2, i.e., we have
\begin{eqnarray}\label{psi-separate}
\ket{\psi}&=&\underbrace{\sum_{\substack{n_1+n_2=N\\\ \\\ }}\psi_{n_1n_20...0}\ \frac{{a_1^\dag}^{n_1}}{\sqrt{n_1!}}\frac{{a_2^\dag}^{n_2}}{\sqrt{n_2!}}\ket{0}}_{\text{\textit{all} particles in modes 1 \& 2}}\ \ +\ \ \underbrace{\sum_{\substack{n_1+...+n_M=N\\n_3+...+n_M\neq0\\\ }}\psi_{n_1n_2...n_M}\ \prod_{i=1}^M\,\frac{{a_i^\dag}^{n_i}}{\sqrt{n_i!}}\ket{0}}_{\text{at least \textit{one} particle in modes 3\,,\,...\,,\,M}}\,.
\end{eqnarray}
Now, consider a scheme with $M-2$ detectors in modes $3,...\,,M$ and post-select on the events without any detection (i.e., all particles remain in modes 1 \& 2). See Fig.~\ref{Fig-InductiveStep} (on the left). This is an event-ready scheme which prepares in modes 1 \& 2 the following state
\begin{eqnarray}\label{post}
\ket{\chi}&\sim&\sum_{{n_1+n_2=N}}\psi_{n_1n_20...0}\ \frac{{a_1^\dag}^{n_1}}{\sqrt{n_1!}}\frac{{a_2^\dag}^{n_2}}{\sqrt{n_2!}}\ket{0}\ \sim\ \frac{1}{\sqrt{N!}}\Big[\,U_1\,a_1^\dag+U_2\,a_2^\dag\,\Big]^N\ket{0}\,,
\end{eqnarray}
where we use the "$\sim$" sign to omit normalisation.
Note that this is the first part in Eq.~(\ref{psi-separate}), corresponding to the post-selected terms ${a_3^\dag}^0...\,{a_M^\dag}^{\!\!\!0}$. Since we require that such obtained state $\ket{\chi}$ \textit{{does not}} show non-local correlations, it must be of the single-mode type, as written on the right for some unitary $\mathbb{U}$. This conclusion comes from the \textit{base case} ($M\!=\!2$), see Eq.~(\ref{psi-M2-single-mode-general-proof}).

This argument restricts the general state in Eq.~(\ref{N-particle-state-adag-general-proof}), which can be written in a simpler form as follows
\begin{eqnarray}\label{psi-separate-Step-1}
\ket{\psi}&=&\gamma\ \frac{1}{\sqrt{N!}}\Big[\,U_1\,a_1^\dag+U_2\,a_2^\dag\,\Big]^N\ket{0}\ +\ \sum_{\substack{n_1+...+n_M=N\\n_3+...+n_M\neq0}}\psi_{n_1n_2...n_M}\ \prod_{i=1}^M\,\frac{{a_i^\dag}^{n_i}}{\sqrt{n_i!}}\ket{0}\,,
\end{eqnarray}
for some normalised vector $(U_1,U_2)$ and parameter $\gamma$ (where the latter is used to restore proper normalisation).

\vspace{0.4cm}
\textsf{\textit{\textbf{\myuline{Part 2}.} {Proof of the inductive step}}.}\vspace{0.3cm}

Consider the following two-mode unitary transformation on modes 1 \& 2
\begin{eqnarray}
\mathbb{U}&:=&\left(\begin{array}{rr}\overline{U}_2&-\overline{U}_1\\U_1&U_2\end{array}\right)\,,
\end{eqnarray}
where $(U_1,U_2)$ is the normalised vector in Eq.~(\ref{psi-separate-Step-1}). It follows that the inverse transformation $\mathbb{U}^{-1}$ applied to the state in Eq.~(\ref{psi-separate-Step-1}) yields
\begin{eqnarray}\label{U-psi-1}
\ket{\tilde{\psi}}\ :=\ \mathbb{U}^{-1}\ket{\psi}&\stackrel{\scriptscriptstyle{(\ref{U-a-dag})}}{=}&\gamma\ \frac{{a_2^\dag}^{N}}{\sqrt{N!}}\ +\sum_{\substack{n_1+...+n_M=N\\n_3+...+n_M\neq0}}\tilde{\psi}_{n_1n_2...n_M}\ \prod_{i=1}^M\,\frac{{a_i^\dag}^{n_i}}{\sqrt{n_i!}}\ket{0}\\\label{U-psi-2}
&=&\gamma\ \frac{{a_2^\dag}^{N}}{\sqrt{N!}}\ +\ \sum_{\bm{{\color{NavyBlue}{k}}}=0}^{N-1}\ \sum_{\substack{n_1+n_3...+n_M=N-\bm{{\color{NavyBlue}{k}}}\\n_3+...+n_M\neq0}}\tilde{\psi}_{n_1\bm{{\color{NavyBlue}{k}}}\,n_2...n_M}\ \frac{{a_1^\dag}^{n_1}}{\sqrt{n_1!}}\frac{\bm{{\color{NavyBlue}{a_2^\dag}^{k}}}}{\sqrt{\bm{{\color{NavyBlue}{k}}}!}}\prod_{i=3}^M\,\frac{{a_i^\dag}^{n_i}}{\sqrt{n_i!}}\ket{0}\,,
\end{eqnarray}
with some new coefficients $\tilde{\psi}_{n_1n_2...n_M}$. Note that because $\mathbb{U}$ acts only on modes 1 \& 2 (and preserves the number of particles), the second sum in Eq.~(\ref{U-psi-1}) continues to collect all terms with at least \textit{one} particle in modes 3\,,\,...\,,\,M. This property is inherited from Eq.~(\ref{psi-separate-Step-1}). For further convenience, the second sum in Eq.~(\ref{U-psi-2}) is reorganised by explicitly pulling out the sum over $\bm{{\color{NavyBlue}{a_2^\dag}^{k}}}$ (which for better visibility are marked in {\color{NavyBlue}{\textbf{bold blue}}}).

Now, consider putting a detector in mode 2 and post-select on a given number of registered particles. See Fig.~\ref{Fig-InductiveStep} (on the right). This results in an even-ready scheme which in the event of $\bm{{\color{NavyBlue}{k}}}=0,1,...\,,N-1$ detections leaves the system prepared in the following respective state
\begin{eqnarray}\label{U-chi-tilde-1}
\ket{\tilde{\chi}_{\bm{{\color{NavyBlue}{k}}}}}&\sim&\sum_{\substack{n_1+n_3+...+n_M=N-\bm{{\color{NavyBlue}{k}}}\\n_3+...+n_M\neq0}}\tilde{\psi}_{n_1\bm{{\color{NavyBlue}{k}}}\,n_3...n_M}\ \frac{{a_1^\dag}^{n_1}}{\sqrt{n_1!}}\prod_{i=3}^M\,\frac{{a_i^\dag}^{n_i}}{\sqrt{n_i!}}\ket{0}\\\label{U-chi-tilde-2}
&\sim&\frac{1}{(N-\bm{{\color{NavyBlue}{k}}})!}\Big[\,\widetilde{U}_1\,a_1^\dag+\widetilde{U}_3\,a_3^\dag+...+\widetilde{U}_M\,a_M^\dag\,\Big]^{N-\bm{{\color{NavyBlue}{k}}}}\ket{0}\,,
\end{eqnarray}
where Eq.~(\ref{U-chi-tilde-1}) readily follows from Eq.~(\ref{U-psi-2}) since post-selection picks out the terms with $\bm{{\color{NavyBlue}{a_2^\dag}^{k}}}$. Note that this state is supported on $M-1$ modes $\{1,3,4,...\,,M\}$. Therefore, we can use the {\textit{inductive assumption}} about the validity of \textbf{Claim~\ref{claim2}} for $M-1$ modes. Form the requirement that the resulting state $\ket{\tilde{\chi}_{\bm{{\color{NavyBlue}{k}}}}}$ is supposed \textit{{not}} to show non-local correlations follows that it must be of the single-mode type. Hence it can be written in the form of Eq.~(\ref{U-chi-tilde-2}) for some $\widetilde{U}_1,\widetilde{U}_3,...\,,\widetilde{U}_M$ (see the  \textbf{Definition~\ref{Def-single-mode-type}} for of a single mode type state supported on modes $1,3,...\,,M$).

A closer look at both Eqs.~(\ref{U-chi-tilde-1}) and (\ref{U-chi-tilde-2}) reveals that 
\begin{eqnarray}\label{no-mode-1}
\tilde{\psi}_{j\,\bm{{\color{NavyBlue}{k}}}\,n_3...n_M}=0&\quad&\text{for \ $j\neq0$}\,,
\end{eqnarray}
and each $\bm{{\color{NavyBlue}{k}}}=0,1,...\,,N-1$.
This can be observed from the following argument (valid for each $\bm{{\color{NavyBlue}{k}}}=0,1,...\,,N-1$):\vspace{0.05cm}

1) There is no term ${a_1^\dag}^{\scriptscriptstyle{N-\bm{{\color{NavyBlue}{k}}}}}$ in Eq.~(\ref{U-chi-tilde-1}) (since for each term in the sum one index $n_3,...\,,n_M$ must be non-zero).
\vspace{0.05cm}

2) There is a term ${a_1^\dag}^{\scriptscriptstyle{N-\bm{{\color{NavyBlue}{k}}}}}$ in Eq.~(\ref{U-chi-tilde-2}), unless $\widetilde{U}_1=0$ (since this term figures in the multinomial expansion).
\vspace{0.05cm}

3) Therefore we must have $\widetilde{U}_1=0$, which entails the lack of any term with ${a_1^\dag}^{\,j}$ in both expressions Eqs.~(\ref{U-chi-tilde-1})/(\ref{U-chi-tilde-2}).
\vspace{0.2cm}

\noindent From Eq.~(\ref{no-mode-1}) we conclude that the state $\ket{\tilde{\psi}}$ in Eq.~(\ref{U-psi-1}) is supported on $M-1$ modes $\{2,3,...\,,M\}$, i.e., there is {\textit{no}} particles in mode 1. In other words, we have
\begin{eqnarray}\label{U-psi-3}
\ket{\tilde{\psi}}&=&\mathbb{U}^{-1}\ket{\psi}\ \stackrel{\scriptscriptstyle{(\ref{no-mode-1})}}{=}\ \gamma\,{a_2^\dag}^{N}\ +\sum_{\substack{n_2+...+n_M=N\\n_3+...+n_M\neq0}}\tilde{\psi}_{0\,n_2...n_M}\ \prod_{i=1}^M\,\frac{{a_i^\dag}^{n_i}}{\sqrt{n_i!}}\ket{0}\,.
\end{eqnarray}
This means that we can use the {\textit{inductive assumption}} again. Since it \textit{{cannot}} reveal non-local correlations, this state is of a single-mode type, i.e., there is a unitary transformation $\mathbb{V}$ such that (cf. Eq.~(\ref{single-mode-def}) for modes $2,3,...\,,M$)
\begin{eqnarray}\label{U-psi-V}
\ket{\tilde{\psi}}&=&\frac{1}{\sqrt{N!}}\Big[\,{\sum}_{i=2}^M\,V_i\,a_i^\dag\,\Big]^N\ket{0}\ =\ \mathbb{V}\ket{0,N,0,...\,,0}\,.
\end{eqnarray}
Therefore, we can write 
\begin{eqnarray}
\ket{\psi}&\stackrel{\scriptscriptstyle{(\ref{U-psi-1})}}{=}&\mathbb{U}\ket{\tilde{\psi}}\ \stackrel{\scriptscriptstyle{(\ref{U-psi-V})}}{=}\ \mathbb{U}\,\mathbb{V}\ket{0,N,0,...\,,0}\ =\ \mathbb{U}\,\mathbb{V}\,\mathbb{S}\ket{N,0,0,...\,,0}\,,
\end{eqnarray}
where $\mathbb{S}$ is the swap operation $a_1^\dag\leftrightarrow a_2^\dag$. This ends the proof of the \textit{inductive step} since the unitary $\mathbb{U}\mathbb{V}$ reduces the state $\ket{\psi}$ to the single-mode mode (for the exact form in Eq.~(\ref{WWW}) take $\mathbb{W}:=\mathbb{U}\mathbb{V}\mathbb{S}$). 

\clearpage

\vspace{0.4cm}
\begin{center}
\textsf{\textbf{D. Proof of Lemma~\ref{equivariance}}}\vspace{0.3cm}
\end{center}

The construction of the local model in the \textbf{Methods} section consists of the description of the hidden variable space Eq.~(\ref{HV-space}) and the action of the basic elements in Eqs.~(\ref{def-beam-splitter})\,-\,(\ref{def-detector}) of passive linear optical setups. The variables are assumed to propagate locally along the paths wired into a complex circuit in which the gates determine evolution. Each path is finally measured providing classical information revealed by detectors, which concludes the experiment. 

Our goal is to describe the observed statistics for a certain class of initial preparations based on the arrangement of the optical circuit. Those initial preparations, associated to the singe-mode type states $\ket{\psi}$ in \textbf{Definition~\ref{Def-single-mode-type}}, are given by the epistemic state
$\mu_\psi({\lambda})$ in Eq.~(\ref{epistemic-state}), i.e.,
\begin{eqnarray}\label{epistemic-state-proof}
\mu_\psi({\lambda})&:=&\tfrac{1}{2\pi}\,\delta_{\vec{\alpha}\sim\vec{\alpha}_\psi}\,\binom{N}{\vec{k}}\,\prod_{i=1}^M\,|{\alpha}_{\scriptscriptstyle{i}}|^{2k_i}\,.
\end{eqnarray}
Let us check that such defined distributions are well-normalised, i.e.,
\begin{eqnarray}
\int_\Lambda d{\lambda}\,\mu_\psi({\lambda})&\stackrel{\scriptscriptstyle{(\ref{HV-space})}}{=}&\int_{\mathbb{C}^M} d\vec{\alpha}\ \tfrac{1}{2\pi}\ \delta_{\vec{\alpha}\sim\vec{\alpha}_\psi}\sum_{\vec{k}\in\mathbb{N}^M}\binom{N}{\vec{k}}\,\prod_{i=1}^M\,|{\alpha}_{\scriptscriptstyle{i}}|^{2k_i}\ =\ \int_0^{2\pi}d\varphi\ \tfrac{1}{2\pi}\ \big(|{\alpha_\psi}_{\scriptscriptstyle{1}}|^2+...+|{\alpha_\psi}_{\scriptscriptstyle{M}}|^2\big)^N\ =\ 1\,,
\end{eqnarray}
where we have used the multinomial expansion and the fact that $\vec{\alpha}_\psi$ is normalised $\Vert\vec{\alpha}_\psi\Vert=1$.

It is straightforward to see that the measurement of the epistemic state
$\mu_\psi({\lambda})$ in Eq.~(\ref{epistemic-state}) by an array of detectors counting the particles, as defined in Eq.~(\ref{def-detector}), complies with the quantum statistics described by Eq.~(\ref{quantum-statistics}), i.e., we have
\begin{eqnarray}\label{HV-statistics}
\mu_\psi({\lambda})&\,\Longrightarrow\,&\text{Prob}_{\mu_\psi({\lambda})}(n_1,...\,,n_M)\ =\ \int_{\mathbb{C}^M} d\vec{\alpha}\ \tfrac{1}{2\pi}\ \delta_{\vec{\alpha}\sim\vec{\alpha}_\psi}\sum_{\vec{k}\in\mathbb{N}^M}\binom{N}{\vec{k}}\,\prod_{i=1}^M\,|{\alpha}_{\scriptscriptstyle{i}}|^{2k_i}\\\nonumber&&\qquad\qquad\qquad\qquad\qquad=\ \binom{N}{\vec{n}}\,\prod_{i=1}^M\,|{\alpha_\psi}_{\scriptscriptstyle{i}}|^{2n_i}\ \stackrel{\scriptscriptstyle{(\ref{probability-single-mode-alpha})}}{=}\ |\psi_{n_1\,...\,n_M}|^2\ =\ \text{Prob}_{\ket{\psi}}(n_1,...\,,n_M)\,.
\end{eqnarray}

The essence of \textbf{Lemma~\ref{equivariance}} is the commutation of the following diagram
\begin{eqnarray}\label{diagram}
\qquad&\begin{matrix}\xymatrix{\ \ \ \ket{\psi}\ \ \ \ar[r]^{}\ar[d]^{} &\ \ \ket{\psi'}\ar[d]^{}\\
\mu_\psi({\lambda})\ \ \ar[r]^{} &\ \ \mu_{\psi'}({\lambda})}
\end{matrix}&
\end{eqnarray}
for each single-mode type state $\ket{\psi}$ and any passive linear optical circuit.
Here, the evolution of $\ket{\psi}\stackrel{\scriptscriptstyle{\mathbb{U}}}{\longrightarrow}\ket{\psi'}=\mathbb{U}\ket{\psi}$ is given by the quantum rules in Eqs.~(\ref{BS}) and (\ref{PS}), whereas the associated epistemic state $\mu_\psi({\lambda})$ evolves according to the definitions in Eqs.~(\ref{def-beam-splitter}) and (\ref{def-phase-shifter}). It means that the shape of the epistemic distribution is preserved (remains within the defined class), and evolves congruently with the quantum state, i.e., we have
\begin{eqnarray}\label{equivariance-proof}
\mu_\psi({\lambda})&\xymatrix{\ar[r]^{\scriptscriptstyle{\mathbb{U}}} &}&\mu'({\lambda})\ =\ \mu_{\psi'}({\lambda})\,,
\end{eqnarray}
where $\ket{\psi}\stackrel{\scriptscriptstyle{\mathbb{U}}}{\longrightarrow}\ket{\psi'}=\mathbb{U}\ket{\psi}$. This will be proved below.

\begin{proof}[\myuline{\textbf{Proof of Eq.~(\ref{equivariance-proof})}}]\ \vspace{0.1cm}\\
\indent The model assumes that the basic elements combine in a modular way, meaning that the gates operate independently and their action is limited to the paths in which they are involved. This simplifies the analysis of a complex optical design, which can be reduced to its most basic elements. It is thus sufficient to prove the property in Eq.~(\ref{equivariance-proof}) for phase shifters and beam splitters. Then, the conclusion regarding arbitrary passive linear optical circuits follows from the modularity of optical designs. 

We consider an epistemic state $\mu_\psi({\lambda})$ in Eq.~(\ref{epistemic-state-proof}) and ask how it transforms under the action of a single gate. Let us start by writing out the stochastic map $\text{Prob}(\lambda'|\lambda)$ which describe the probability that the variable $\lambda'$ evolves to $\lambda$, i.e. $\lambda'\longrightarrow\lambda$. For the respective gates, as defined in Eq.~(\ref{def-phase-shifter}) and (\ref{def-beam-splitter}), it takes the following form
\begin{eqnarray}\label{stochastic-map-PS}
\text{for \textit{phase shifter}:\,}&&\quad\text{Prob}(\lambda|\lambda')\ =\ \delta_{\vec{\alpha}\text{=}\,\vec{\alpha}'\mathbb{U}}\ \delta_{\vec{k}\text{=}\,\vec{k}'}\ ,
\\\label{stochastic-map-BS}
\text{for \textit{beam splitter}:}&&\quad\text{Prob}(\lambda|\lambda')\ =\ \delta_{\vec{\alpha}\text{=}\,\vec{\alpha}'\mathbb{U}}\prod_{j\neq s,t}\delta_{k_j\text{=}\,k_j'}\begin{pmatrix}k_s'+k_t'\\k_s\,,k_t\end{pmatrix}\tfrac{|\alpha_s|^{2k_s}\,|\alpha_t|^{2k_t}}{(|\alpha_s|^2+|\alpha_t|^2)^{k_s'+k_t'}}\ ,
\end{eqnarray}
where the unitary $\mathbb{U}$ is associated with the respective gate as described in Eqs.~(\ref{PS}) and (\ref{BS}).

The \textit{phase shifter} is a deterministic gate for which we get
\begin{eqnarray}\label{equivariance-proof-PS}
\mu_\psi({\lambda})&\!\!\xymatrix{\ar[r]^{} &}\!\!&\mu'({\lambda})\ =\ \int_\Lambda d{\lambda'}\,\text{Prob}(\lambda|\lambda')\,\mu_\psi({\lambda'})\ \stackrel{\scriptscriptstyle{(\ref{epistemic-state-proof})(\ref{stochastic-map-PS})}}{=}\ \tfrac{1}{2\pi}\,\delta_{\vec{\alpha}\sim\vec{\alpha}_\psi\mathbb{U}}\,\binom{N}{\vec{k}}\,\prod_{i=1}^M\,|{\alpha_i}|^{2k_i}\ =\ \mu_{\psi'}({\lambda})
\,,
\end{eqnarray}
where $\vec{\alpha}_{\psi'}=\vec{\alpha}_\psi\mathbb{U}$ corresponds to the state $\ket{\psi'}=\mathbb{U}\ket{\psi}$ transformed by the phase shifter as described by Eq.~(\ref{PS}); see  \textbf{Proposition~\ref{alpha-transformation}}. This proves the equivariance Eq.~(\ref{equivariance-proof}) for phase shifters.

The \textit{beam splitter} is a stochastic gate which leads to the following evolution
\begin{eqnarray}\label{equivariance-proof-BS-1}
\!\!\!\!\!\!\!\!\!\!\!\!\!\!\mu_\psi({\lambda'})&\!\!\xymatrix{\ar[r]^{} &}\!\!&\mu'({\lambda})\ =\ \int_\Lambda d{\lambda'}\ \text{Prob}(\lambda|\lambda')\,\mu_\psi({\lambda'})\ =\ \int_{\mathbb{C}^M} d\vec{\alpha}'\sum_{\vec{k}'\in\mathbb{N}^M}\text{Prob}(\lambda|\lambda')\,\mu_\psi({\lambda'})
\\\label{equivariance-proof-BS-2}
&&\qquad\ \stackrel{\scriptscriptstyle{(\ref{epistemic-state-proof})(\ref{stochastic-map-BS})}}{=}\ \tfrac{1}{2\pi}\int_{\mathbb{C}^M} d\vec{\alpha}'\  \delta_{\vec{\alpha}\text{=}\,\vec{\alpha}'\mathbb{U}}\ \delta_{\vec{\alpha}'\sim\vec{\alpha}_\psi}\sum_{\vec{k}'\in\mathbb{N}^M}
\prod_{j\neq s,t}\delta_{k_j\text{=}\,k_j'}\begin{pmatrix}k_s'+k_t'\\k_s\,,k_t\end{pmatrix}\tfrac{|\alpha_s|^{2k_s}\,|\alpha_t|^{2k_t}}{(|\alpha_s|^2+|\alpha_t|^2)^{k_s'+k_t'}}\ 
\binom{N}{\vec{k}'}\,\prod_{i=1}^M\,|{\alpha}_{\scriptscriptstyle{i}}|^{2k_i'} 
\\\label{equivariance-proof-BS-3}
&&\qquad\ \ \ \,=\ \tfrac{1}{2\pi}\,\delta_{\vec{\alpha}\sim\vec{\alpha}_\psi\mathbb{U}}\sum_{\vec{k}'\in\mathbb{N}^M}
\prod_{j\neq s,t}\delta_{k_j\text{=}\,k_j'}\begin{pmatrix}k_s'+k_t'\\k_s\,,k_t\end{pmatrix}\tfrac{|\alpha_s|^{2k_s}\,|\alpha_t|^{2k_t}}{(|\alpha_s|^2+|\alpha_t|^2)^{k_s'+k_t'}}\ 
\binom{N}{\vec{k}'}\,\prod_{i=1}^M\,|\alpha_i|^{2k_i'}
\\\label{equivariance-proof-BS-4}
&&\qquad\ \ \ \,=\ \tfrac{1}{2\pi}\,\delta_{\vec{\alpha}\sim\vec{\alpha}_\psi\mathbb{U}}\,\binom{N}{\vec{k}}\,\prod_{i=1}^M\,|\alpha_i|^{2k_i}\sum_{k_s',k_t'}
\begin{pmatrix}k_s'+k_t'\\k_s\,,k_t\end{pmatrix}\tfrac{|\alpha_s|^{2k_s'}\,|\alpha_t|^{2k_t'}}{(|\alpha_s|^2+|\alpha_t|^2)^{k_s'+k_t'}}\ 
\tfrac{k_s!k_t!}{k_s'!k_t'!}\ ,
\end{eqnarray}
where in the last line we have reshuffled the terms and carried out all the Kronecker deltas $\prod_{j\neq s,t}\delta_{k_j\text{=}\,k_j'}$.

Now, observe that the second binomial in Eq.~(\ref{equivariance-proof-BS-4}) is non-vanishing only for $k_s+k_t=k_s'+k_t'$. This simplifies the above expression as follows
\begin{eqnarray}\label{equivariance-proof-BS-5}
\qquad\qquad\quad\mu'({\lambda})&=&\tfrac{1}{2\pi}\,\delta_{\vec{\alpha}\sim\vec{\alpha}_\psi\mathbb{U}}\,\binom{N}{\vec{k}}\,\prod_{i=1}^M\,|\alpha_i|^{2k_i}\sum_{k_s',k_t'}
\begin{pmatrix}k_s+k_t\\k_s'\,,k_t'\end{pmatrix}\tfrac{|\alpha_s|^{2k_s'}\,|\alpha_t|^{2k_t'}}{(|\alpha_s|^2+|\alpha_t|^2)^{k_s+k_t}}\\\label{equivariance-proof-BS-6}
&=&\tfrac{1}{2\pi}\,\delta_{\vec{\alpha}\sim\vec{\alpha}_\psi\mathbb{U}}\,\binom{N}{\vec{k}}\,\prod_{i=1}^M\,|\alpha_i|^{2k_i}\,\tfrac{(|\alpha_s|^2+|\alpha_t|^2)^{k_s+k_t}}{(|\alpha_s|^2+|\alpha_t|^2)^{k_s+k_t}}\ =\ \tfrac{1}{2\pi}\,\delta_{\vec{\alpha}\sim\vec{\alpha}_\psi\mathbb{U}}\,\binom{N}{\vec{k}}\,\prod_{i=1}^M\,|\alpha_i|^{2k_i}\ =\ \mu_{\psi'}({\lambda})\,,
\end{eqnarray}
where $\vec{\alpha}_{\psi'}=\vec{\alpha}_\psi\mathbb{U}$ corresponds to the state $\ket{\psi'}=\mathbb{U}\ket{\psi}$ transformed by the beam splitter as described by Eq.~(\ref{BS}); see \textbf{Proposition~\ref{alpha-transformation}}. 
This proves the equivariance Eq.~(\ref{equivariance-proof}) for beam splitters.

\end{proof}

\end{document}